\documentclass[11pt,reqno]{amsart}

\pdfoutput = 1

\usepackage[english]{babel}
\usepackage{amsmath,amsfonts,amssymb,amsthm}
\usepackage{amsfonts}
\usepackage{amsxtra}
\usepackage{array,dcolumn}
\usepackage{graphicx}
\usepackage[hyperfootnotes=true,
        pdffitwindow=true,
        plainpages=false,
        pdfpagelabels=true,
        pdfpagemode=UseOutlines,
        pdfpagelayout=SinglePage,
        hyperindex,]{hyperref}

\usepackage[T1]{fontenc}
\usepackage[utf8]{inputenc}
\usepackage[activate={true,nocompatibility},spacing,kerning]{microtype}
\usepackage[sc,osf]{mathpazo}
\microtypecontext{spacing=nonfrench}

 \calclayout

\newcommand{\mathsym}[1]{{}}
\newcommand{\unicode}[1]{{}}

\numberwithin{equation}{section}

\theoremstyle{plain}
\newtheorem{theorem}{Theorem}

\newtheorem{corollary}[theorem]{Corollary}
\newtheorem{proposition}[theorem]{Proposition}
\numberwithin{theorem}{section}

\theoremstyle{definition}

\theoremstyle{remark}
\newtheorem{remark}[theorem]{Remark}

\newcommand{\+}{\!+\!}

\renewcommand{\leq}{\leqslant}
\renewcommand{\geq}{\geqslant}

\newcommand{\hypergeometric}[6][\bigg]{\,{}_{#2} F_{#3} #1
( \begin{matrix} #4 \\ #5 \end{matrix}\, #1\vert\, #6 #1)}

\begin{document}
\title[Jacobi and Cauchy ensemble moments]{Relations between moments for the Jacobi and Cauchy random matrix ensembles}
\author{Peter J. Forrester}
\address{School of Mathematics and Statistics, 
ARC Centre of Excellence for Mathematical
 and Statistical Frontiers,
University of Melbourne, Victoria 3010, Australia}
\email{pjforr@unimelb.edu.au}

\author{Anas A. Rahman}
\address{School of Mathematics and Statistics, 
ARC Centre of Excellence for Mathematical
 and Statistical Frontiers,
University of Melbourne, Victoria 3010, Australia}
\email{anas.rahman@live.com.au}

\begin{abstract}
We outline a relation between the densities for the $\beta$-ensembles with respect to the Jacobi weight $(1-x)^a(1+x)^b$ supported on the interval $(-1,1)$ and the Cauchy weight $(1-\mathrm{i}x)^{\eta}(1+\mathrm{i}x)^{\bar{\eta}}$ by appropriate analytic continuation. This has the consequence of implying that the latter density satisfies a linear differential equation of degree three for $\beta=2$, and of degree five for $\beta=1$ and $4$, analogues of which are already known for the Jacobi weight $x^a(1-x)^b$ supported on $(0,1)$. We concentrate on the case $a=b$ (Jacobi weight on $(-1,1)$) and $\eta$ real (Cauchy weight) since the density is then an even function and the differential equations simplify. From the differential equations, recurrences can be obtained for the moments of the Jacobi weight supported on $(-1,1)$ and/or the moments of the Cauchy weight. Particular attention is paid to the case $\beta=2$ and the Jacobi weight on $(-1,1)$ in the symmetric case $a=b$, which in keeping with a recent result obtained by Assiotis et al.~for the $\beta=2$ case of the symmetric Cauchy weight (parameter $\eta$ real), allows for an explicit solution of the recurrence in terms of particular continuous Hahn polynomials. Also for the symmetric Cauchy weight with $\eta=-\beta(N-1)/2-1-\alpha$, after appropriately scaling $\alpha$ proportional to $N$, we use differential equations to compute terms in the $1/N^2$ ($1/N$) expansion of the resolvent for $\beta=2$ ($\beta=1,4$).
\end{abstract}
\maketitle

\section{Introduction}\label{s1}
The study of moments in random matrix theory was initiated by Wigner in 1955 \cite{Wi55}. This was in the context of analysing the limiting
spectral density of a particular class of real symmetric matrices $J_N$ of size $(2N+1)\times (2N+1)$, with diagonal entries set to zero, and the
strictly upper triangular entries independently chosen from the values $\pm w$ with equal probability. Upon scaling the eigenvalues by
$1/\sqrt{N}$, or equivalently multiplying $J_N$ by this value to obtain the scaled matrices $\bar{J}_N$, Wigner considered the trace of
even powers of $\bar{J}_N$, and was able to prove that for $p \in \mathbb Z_{\ge 0}$,
\begin{equation}\label{1.1}
\lim_{N \to \infty} {1 \over N}  \Big \langle {\rm Tr} \, \bar{J}_N^{2p} \Big \rangle  = w^p {1 \over p + 1} \binom{2p}{p}.
\end{equation}
The combinatorial expression on the RHS is recognised as the $p$-th Catalan number. 

The average of the LHS of (\ref{1.1}) is an example of a linear statistic of the eigenvalues $\{\lambda_j\}_{j=1}^N$.
Indeed, $  {\rm Tr} \, \bar{J}_N^{2p}  = \sum_{j=1}^N \lambda_j^{2p}$ is a sum over individual eigenvalues, which is
the defining feature of a linear statistic. Introduce the scaled eigenvalue density $\bar{\rho}^{(N)}(\lambda)$  corresponding to the
 matrices $\bar{J}_N$, which by definition has the property that $N\int_a^b \bar{\rho}^{(N)}(\lambda) \, d\lambda$ is equal to the 
expected number of eigenvalues in the interval $[a,b]$. The fact that  $  {\rm Tr} \, \bar{J}_N^{2p} $ is a linear statistic tells us
that its average can be expressed as an integral over the scaled eigenvalue density,
\begin{equation}\label{1.2}
 {1 \over N}  \Big \langle {\rm Tr} \, \bar{J}_N^{2p} \Big \rangle  =  \int_{-\infty}^\infty \lambda^{2p} \bar{\rho}^{(N)}(\lambda)\, d\lambda.
 \end{equation}
 Consider now the particular density function $\rho^{\rm W}(\lambda)$ supported on $[-2w,2w]$ specified by the functional form
\begin{equation}\label{1.3} 
\rho^{\rm W}(\lambda)=\frac{1}{2\pi w^2}(4w^2-\lambda^2)^{1/2}.
 \end{equation}
 This is referred to as the Wigner semi-circle. A simple change of variables and use of  knowledge of the Euler
 beta function in terms of gamma functions verifies that
 \begin{equation}\label{1.4} 
 \int_{-2w}^{2w} \lambda^{2p} \rho^{\rm W}(\lambda) \, d \lambda =  w^p {1 \over p + 1} \binom{2p}{p}.
 \end{equation}
 Thus upon comparing to (\ref{1.1}) and recalling (\ref{1.2}), we have
  \begin{equation}\label{1.5} 
\lim_{N \to \infty}    \int_{-\infty}^\infty \lambda^{2p} \bar{\rho}^{(N)}(\lambda)\, d\lambda =  \int_{-2w}^{2w} \lambda^{2p} \rho^{\rm W}(\lambda) \, d \lambda.
 \end{equation}
 
 In relation to (\ref{1.4}), and using too the fact that the odd moments vanish, the exponential generating function
 can be formed to obtain
  \begin{equation}\label{1.6x} 
   \int_{-2w}^{2w} e^{ i s \lambda} \rho^{\rm W}(\lambda) \, d \lambda =   {J_1(2w s) \over ws},
 \end{equation}
where $J_\nu$ denotes the Bessel function of order $\nu$. This is analytic at the origin, 
telling us that $ \rho^{\rm W}(\lambda) $ is uniquely defined by its moments. This fact is sufficient
for the limit formula (\ref{1.5}) to imply that the sequence of measures $ \{\bar{\rho}^{(N)}(\lambda)\, d\lambda \}_{N=1}^\infty$ 
converges weakly to $ \rho^{\rm W}(\lambda) \, d \lambda$, i.e.~to the Wigner semi-circle. This was the conclusion of Wigner's
analysis. 
Wigner's second paper on this topic, published in 1958 \cite{Wi58}, 
considered a wider class of random real symmetric matrices. Independence of entries, fixed variance and bounded moments were shown, by the same technique, to be sufficient conditions
for the limiting spectral density to be given by the Wigner semi-circle.

Knowledge of  the limiting form of the spectral density and moments is refined by knowledge of finite $N$
corrections. For the class of Wigner matrices, i.e.~real symmetric or complex Hermitian random matrices
 with entries on the diagonal chosen independently from a particular zero mean distribution, and
 upper triangular entries chosen independently from a finite mean, finite variance distribution, it is 
 known that averages of suitable test functions $\phi(\lambda)$ over $ \bar{\rho}^{(N)}(\lambda)$ permit $1/N$ expansions
   \begin{equation}\label{1.6} 
 \int_{-\infty}^\infty \phi(\lambda)\bar{\rho}^{(N)}(\lambda)\, d\lambda 
	= \int_{-2w}^{2w} \phi(\lambda)\rho^{\rm W,0}(\lambda) \, d\lambda 
	+ \frac{1}{N} \int_{-\infty}^\infty \phi(\lambda)\rho^{\rm W,1}(\lambda) \, d\lambda
	+ \cdots
\end{equation}
where $\rho^{\rm W,0}(\lambda) = \rho^{\rm W}(\lambda)$ is the Wigner semi-circle (\ref{1.3}) and $\rho^{\rm W,1}(\lambda)$
depends on the second moment of the diagonal entries, and moments up to and including the fourth
of the off-diagonal elements \cite{BY05, KKP96}; see the recent work \cite{FT19c} for a convenient
summary.

The explicit form of the expansion (\ref{1.6}) is known to higher orders in the case of real symmetric Wigner matrices
with independent normals ${\rm N}[0,\sqrt{2} w]$ on the diagonal, and ${\rm N}[0,w]$ on the off-diagonal (GOE matrices),
or complex Hermitian Wigner matrices with  independent normals ${\rm N}[0,w]$ on the diagonal, and
${\rm N}[0, w/\sqrt{2}] + i{\rm N}[0, w/\sqrt{2}] $ on the off-diagonal (GUE matrices); see \cite{WF14}.  In general, 
for a class of scaled Wigner matrices $\{X \}$ with symmetric distributions so that the odd moments vanish,
the $2p$-th power of the trace is a  polynomial in $1/N$ of degree $p$. For the GOE and GUE this polynomial ---
the constant term of which is given by the RHS of (\ref{1.1}) --- has a number of unique features among the
class of Wigner matrices which can be traced back to the fact that the spectral density in these cases
can be evaluated in terms of Hermite polynomials; see \cite{MS11} and \cite{WF14} as well as the earlier
works \cite{HT03} and \cite{Le04}. One of these features is that the moments can be characterised
by a recurrence. The simplest case is the scaled GUE (we take $w = 1/2$ for definiteness), for which a result first obtained by Ullah
\cite{Ul85} tells us that
  \begin{equation}\label{1.7}
\int_{-\infty}^\infty e^{i k \lambda} \bar{\rho}^{(N)}(\lambda) \, d \lambda =
   {1 \over N} e^{-k^2/8N} L_{N-1}^{(1)}(k^2/4N),
   \end{equation} 
   where $L^{(p)}_n(x)$ denotes the Laguerre polynomial. Known asymptotics
   of the Laguerre polynomials shows that (\ref{1.7}) limits to (\ref{1.6x}), as it must.
   As noted in the recent work \cite{Fo20}, this explicit formula can be used
   to show that the LHS satisfies a second order differential equation
   \cite{HT03,Le04}, and that $ \bar{\rho}^{(N)}(\lambda)$ satisfies a third order
   differential equation \cite{LM79,GT05}. From the former a second recurrence satisfied
   by the sequence of $2p$-th moments can be deduced \cite{HZ86,HT03}.
   Analogous results can be obtained for the GOE but with a higher level of complexity.
   The LHS of (\ref{1.7}) now satisfies a fourth order linear differential equation, and the
   sequence of $2p$-th moments satisfies a fourth order difference equation
   \cite{Le09}. The density satisfies a fifth order differential equation
   \cite{WF14}.
   
   The Hermite polynomials expressing the spectral density for the GOE and GUE
   (and too the GSE --- see e.g.~\cite[Ch.~6]{Fo10}) correspond to the Gaussian
   weight in the theory of classical orthogonal polynomials. We recall that a weight function
   is said to be classical if its logarithmic derivative ${d \over dx} \log w(x)$ can be written
   in the form $-g(x)/ f(x)$ with $f>0$, $f$ and $g$ having no common factors,
   and the degree of $f$ ($g$) less than or equal to $2$ $(1)$; see e.g.~\cite[\S 5.4.3]{Fo10}.
   Up to fractional linear transformations, it is known that the only classical weights
   with support on the real line are the Gaussian, the Laguerre weight
   $x^a e^{-x} \chi_{x > 0}$ supported on $x>0$, as well as the Jacobi and Cauchy weights 
    \begin{equation}\label{1.8} 
    (1 - x)^a (1 + x)^b \chi_{-1 < x < 1}, \qquad (1 - ix)^\eta (1 + ix )^{\bar{\eta}}.
    \end{equation} 
    The indicator function $\chi_A$ is defined to equal $1$ if $A$ is true and $0$ otherwise.
    Note that a simple linear change of variables maps the Jacobi weight
    as presented in (\ref{1.8}) to 
     \begin{equation}\label{1.8a} 
     x^a (1 - x)^b \chi_{0<x<1}.
     \end{equation}
    
    Eigenvalue probability density functions (PDFs) proportional to
     \begin{equation}\label{1.9}  
     \prod_{l=1}^N w(x_l) \prod_{1 \le j < k \le N} |x_k - x_j|^\beta,
  \end{equation}      
with    $w(x)$ a classical weight, and $\beta =1,2$ or $4$, are prominent
in classical random matrix theory. For example, with the choice
of $w(x)$ as a Gaussian, (\ref{1.9}) is the eigenvalue PDF
for the GOE when $\beta = 1$, the GUE when $\beta = 2$ and
the GSE when $\beta = 4$. Thus the exponent $\beta$ --- often referred
to as the Dyson index after the pioneering work \cite{Dy62c} --- corresponds
to the number of independent real parts in the corresponding number field.
Beyond these three special values of $\beta$,
for the classical weights there are constructions of random matrices with
eigenvalue PDF (\ref{1.9}) for general $\beta > 0$. In the case of the Jacobi
weight as written in (\ref{1.8}), this was first obtained in \cite{KN04}, while
for the Cauchy weight it was obtained in \cite{FR06}; for a text
book treatment see \cite[\S 3.11 \& \S 4.3.4]{Fo10}.
For the Gaussian, Laguerre and Jacobi weights, the theory of Selberg
correlation integrals as applied in \cite{RF19} tells us that there are integrable
structures by way of linear differential equations 
for the density of degree $\beta + 1$ for all $\beta$ even, and that duality
formulas extend this characterisation when $\beta$ is replaced by $4/\beta$.
From the differential equations, difference equations for the moments can be determined.
In the Jacobi case these are of the same degree, but in the Gaussian and
Laguerre cases their degree reduces by one to now be equal to $\beta$. In practice, beyond
the classical value $\beta =2,4$ (and by the duality, $\beta = 1$) it was not
feasible to make these differential or difference equations explicit. An exception
was $\beta = 6$ in the Gaussian case (and by duality $\beta = 2/3$), where the
seventh order differential equation, and sixth order difference equation were presented
explicity.

Absent from the study \cite{RF19} was consideration of analogous integrable
properties in the case of the Cauchy weight in (\ref{1.9}); see \cite[\S2.5, \S 3.9, \S 4.3.4]{Fo10} for
a textbook treatment relating to the latter.
 With $\beta=2$,  results of the
type obtained in \cite{RF19} for the Jacobi weight (\ref{1.8a}) (see also
\cite{Le04,MS11,CMOS18}) have recently been 
obtained for the Cauchy weight in the work of Assiotis et al.~\cite{ABGS20}, with the requirement that
the Cauchy parameter $\eta$ in (\ref{1.8}) be real.
In particular, it was shown that the sum of the $2p$-th and
$(2p+2)$-th moment can be identified as a continuous Hahn
polynomial in the variable $p$. In \cite{CMOS18} it has earlier
been shown that the difference between successive moments in
the Jacobi case of (\ref{1.9}) can be identified in terms of Wilson
polynomials from the Askey scheme (for other recent appearances of
polynomials from the Askey scheme in studies in random matrix
theory see \cite{FL19,FL19a,GGR20}).

Underpinning the present paper is the observation that integrations with
respect to the Cauchy ensemble can be computed by an analytic continuation
in the parameters $a,b$ of integrations in the Jacobi ensemble with weight defined
on $(-1,1)$. The details of such inter-relations are given in Section \ref{S2}.
From our work \cite{RF19} giving linear differential equations for the density
of the Jacobi ensemble defined on $(0,1)$ with $\beta =1$, $2$ and $4$,
a simple linear change of variables gives differential equations for
the corresponding Jacobi ensemble defined on $(-1,1)$. Applying the results of
Section \ref{S2} then gives us  linear differential equations for the density
of the Cauchy ensemble with $\beta =1$, $2$ and $4$, which we present
in Section \ref{S3}. We highlight the symmetric case $\eta=\bar{\eta}$ and take $\eta$ to be general complex only in the $\beta=2$ case in order to keep our presentation neat. In the $\beta = 2$ case with $\eta$ real,
 we reclaim the third order differential equation derived recently in \cite{ABGS20}. Moreover, the differential equations for the densities of the $\beta=1,2,4$ Cauchy ensembles simplify if written with  $(1+x^2) \rho_{(1)}^{(Cy)}(x)$
as the dependent variable;
in the Jacobi case, the inter-relations tell us that the analogous simplification
occurs with $(1-x^2) \rho_{(1)}^{(J)}(x)$ as the dependent variable. 

Two consequences are further developed. One is to quantify the limiting density and its moments in the symmetric
Cauchy ensemble with $\eta = - \beta (N - 1)/2 - 1- \hat{\alpha}\beta N/2$. It is also
possible to use differential equations to study $1/N^2$ corrections (for $\beta = 2$)
and $1/N$ corrections (for $\beta = 1,4$) of the type well known in the study of the GUE \cite{HZ86}. The other is the specification of (three-term for $\beta = 2$, five-term for $\beta = 1,4$) recurrences for the differences of successive even moments of the Jacobi ensembles on $(-1,1)$, which are equivalent (up to minus signs) to recurrences for the sums of successive even moments of the corresponding Cauchy ensembles. These recurrences are presented in Section \ref{S4}, along with a demonstration of the fact that the recurrence in the $\beta=2$, $a=b$ Jacobi case on $(-1,1)$ can be solved in terms of continuous Hahn polynomials. This latter fact is consistent, via the inter-relations of Section \ref{S2}, to an equivalent observation of \cite{ABGS20} on the associated recurrence in the Cauchy case.

\section{Relating the density for the Cauchy and Jacobi $\beta$-ensembles}\label{S2}
The density in the Cauchy $\beta$-ensemble is specified by
\begin{equation}\label{2.1}
\rho_{(1)}^{(Cy)}(x) = N {w_\beta^{(Cy)}(x) \over \mathcal N_N^{(Cy)}}
\int_{-\infty}^\infty dx_2 \cdots \int_{-\infty}^\infty dx_N \,
\prod_{l=2}^N w_\beta^{(Cy)}(x_l)| x - x_l |^\beta  \prod_{2 \le j < k \le N} | x_k - x_j |^\beta,
\end{equation}
where
\begin{equation}\label{2.2}
w_\beta^{(Cy)}(x) = (1 - i x)^{\alpha_{\beta,N}} (1 + i x)^{\bar{\alpha}_{\beta,N}} , \qquad
{\alpha}_{\beta,N} = - \beta (N - 1)/2 - 1 - \alpha.
\end{equation}
Here, $\alpha \in \mathbb C$ and it is required that ${\rm Re} \, \alpha > - 1/2$. The dependence
on $\beta$ and $N$ in the exponent $\alpha_{\beta,N}$ has been chosen so that this latter
requirement ensures that the normalisation $ \mathcal N_N^{(Cy)}$ is finite; see \cite[\S 3.9]{Fo10}
for working which gives insight into this convergence condition. In fact, the normalisation is
known explicitly (see \cite[Exercises 4.7 q.4(i)]{Fo10}),
\begin{equation}\label{2.2a}
\mathcal N_N^{(Cy)} =  2^{-\beta N(N-1)/2 - 2N{\rm Re} \, \alpha }
\pi^N \,M_N(\alpha,\bar{\alpha},\beta/2),
\end{equation}
where
\begin{equation}\label{2.2b}
M_N(a,b,\lambda) = \prod_{j=0}^{N - 1} { \Gamma (\lambda j+a+b+1)
\Gamma(\lambda (j+1)+1) \over
  \Gamma (\lambda j+a+1)\Gamma (\lambda j+b+1) \Gamma (1 + \lambda)}.
\end{equation}  

In relation to the Jacobi weight from (\ref{1.8}), the density for the corresponding
Jacobi $\beta$-ensemble is
\begin{equation}\label{2.3}
\rho_{(1)}^{(J)}(x) = N {w_\beta^{(J)}(x) \over \mathcal N_N^{(J)}}
\int_{-1}^1 dx_2 \cdots \int_{-1}^1 dx_N \,
\prod_{l=2}^N w_\beta^{(J)}(x_l) | x - x_l |^\beta  \prod_{2 \le j < k \le N} | x_k - x_j |^\beta, 
\end{equation}
where
\begin{equation}\label{2.4}
 w_\beta^{(J)}(x) = (1- x)^a (1 + x)^b\chi_{-1<x<1}.
 \end{equation}
 Here, the normalisation is finite for ${\rm Re} \, a,b > -1$ and is given in terms of 
 the Selberg integral (see \cite[\S 4.1]{Fo10})
 \begin{equation}\label{2.5}
  \mathcal N_N^{(J)} =  2^{ \beta N (N - 1)/2 + N (a + b+1)}  \,S_N(a,b,\beta/2),
 \end{equation} 
 where
  \begin{equation}\label{2.6}
 S_N(a,b,\lambda) =  \prod_{j=0}^{N-1} {\Gamma (a + 1 + \lambda j)
\Gamma (b + 1 + \lambda j)\Gamma(1+\lambda(j+1)) \over
\Gamma (a+b + 2 +\lambda (N + j-1)) \Gamma (1 + \lambda )}.
 \end{equation} 
 
  \subsection{The symmetric case}
 
 Let us now specialise the Cauchy weight to the case $\alpha$ real, and specialise
 the Jacobi weight to the case $a=b$. Both weights are then even functions of
 $x$. The following relation between multiple
 integrals over these weights holds true:
 
 \begin{proposition}\label{P2.1}
 Let $f(x_1,\dots,x_N)$ be a multivariable symmetric polynomial of degree $d$ in each
 $x_i$. For $2 \eta < -(d +1)$, define
 \begin{equation}\label{2.7} 
I_{N,\eta}^{(Cy)}[f(x_1,\ldots,x_N)] := \int_{-\infty}^\infty dx_1 \, (1 + x_1^2)^\eta \cdots \int_{-\infty}^\infty dx_N \, (1 + x_N^2)^\eta \,
 f(x_1,\dots,x_N),
 \end{equation}
 and for $\eta$ outside of this range, define $ I_{N,\eta}^{(Cy)}[f(x_1,\ldots,x_N)] $ by its analytic continuation. Also,
 in relation to the Jacobi weight with $a=b > -1$, define
  \begin{equation}\label{2.8} 
I_{N,a}^{(J)}[f(x_1,\ldots,x_N)] := \int_{-1}^1 dx_1 \, (1 - x_1^2)^a \cdots \int_{-1}^1 dx_N \, (1 - x_N^2)^a \,
 f(x_1,\dots,x_N),
 \end{equation} 
 and for $a$ outside of this range, define $ I_{N,a}^{(J)}[f(x_1,\ldots,x_N)]$ by its analytic continuation.
 We have
   \begin{equation}\label{2.9} 
  I_{N,\eta}^{(Cy)}[f(ix_1,\dots i x_N)] = (\tan \pi \eta)^N    I_{N,\eta}^{(J)}[f(x_1,\dots  x_N)].
  \end{equation} 
  \end{proposition}
  
  \begin{proof}
  Let $p \in \mathbb Z_{\ge 0}$. Suppose $2 \eta < - (p+1)$ and $a>-1$. A simple
  change of variables (for $k$ even)
  and use of the Euler beta integral evaluation (see \cite[Exercises 5.4 q.2]{Fo10})  shows
  \begin{equation}\label{6.2a}   
  \int_{-\infty}^\infty  (1 + x^2)^\eta x^{k} \, dx =
  \begin{cases} 0, & k \: {\rm odd}, \\
  (-1)^{k/2} \tan \pi \eta \, {\Gamma(1 + \eta) \Gamma((k+1)/2) \over
  \Gamma((k+3)/2 + \eta)}, & k \: {\rm even},
  \end{cases}
   \end{equation} 
  and
  \begin{equation}\label{6.2b}   
   \int_{-1}^1  (1 - x^2)^a x^{k} \, dx =
  \begin{cases} 0, & k \: {\rm odd}, \\
   \, {\Gamma(1 + a) \Gamma((k+1)/2) \over
  \Gamma((k+3)/2 + a)}, & k \: {\rm even}.
  \end{cases}
   \end{equation}  
 
 The functions $f(x_1,\ldots,x_N)$ and $f(ix_1,\ldots,ix_N)$ are polynomials, so the computation of $ I_{N,\eta}^{(Cy)}[f(ix_1,\ldots,ix_N)] $
 and $  I_{N,a}^{(J)}[f(x_1,\ldots,x_N)] $ reduces to the above one-dimensional integrals.
 Since, as analytic functions of $\eta$, we read off from the respective evaluations that
 $$
 \int_{-\infty}^\infty  (1 + x^2)^\eta (ix)^{2k} \, dx =    \tan \pi \eta   \int_{-1}^1  (1 - x^2)^\eta x^{2k} \, dx, 
 $$
 the stated result (\ref{2.9}) follows
 \end{proof}
 
 One immediate consequence is a relation between the normalisations in
 (\ref{2.1}) and (\ref{2.3}) in specialisations of the parameters that conform
 with Proposition \ref{P2.1}.
 
 \begin{corollary}\label{C2.2}
 Let $\alpha$ be real and related to $\alpha_{\beta,N}$ as in (\ref{2.2}). For $\beta$ even, we have
 \begin{equation}\label{2.10}
(-1)^{\beta N (N - 1) / 4} \mathcal N_N^{(Cy)} = ( -\tan  \pi \alpha )^N  \mathcal N_N^{(J)} \Big |_{a = b = - \beta (N - 1)/2 - 1 - \alpha},
 \end{equation}
 where both sides are to be interpreted as analytic functions in $\alpha$.
 \end{corollary}
 
 \begin{proof}
 With $\alpha$ and thus $\alpha_{\beta,N}$ real, we have
 $$
 \mathcal N_N^{(Cy)} = \int_{-\infty}^\infty dx_1 (1 + x_1^2)^{\alpha_{\beta,N}} \cdots
   \int_{-\infty}^\infty dx_N (1 + x_N^2)^{\alpha_{\beta,N}} \prod_{1 \le j < k \le N} | x_k - x_j|^\beta.
   $$
   Also, with $a=b$, we have
   $$
 \mathcal N_N^{(J)} =    \int_{-1}^1 dx_1 \, (1 - x_1^2)^a \cdots   \int_{-1}^1 dx_N \, (1 - x_N^2)^a 
 \prod_{1 \le j < k \le N} | x_k - x_j|^\beta.
 $$
 For $\beta$ even, $ \prod_{1 \le j < k \le N} | x_k - x_j|^\beta$ is a multivariable symmetric polynomial, so
 application of Proposition \ref{P2.1} gives (\ref{2.10}).
 \end{proof}
 
 \begin{remark}
 The explicit form of the analytic continuations in $\alpha$ of both sides of (\ref{2.10})
 is known from (\ref{2.2a}) and (\ref{2.5}). In the notation therein, the equality
 (\ref{2.10}) requires
  \begin{multline}\label{2.11}
 (-1)^{\beta N (N - 1) / 4} (2 \pi )^N M_N(\alpha, \alpha,\beta/2) \\
 = (-\tan \pi \alpha)^N S_N(-\beta(N-1)/2-1-\alpha,
  -\beta (N - 1)/2 - 1 - \alpha, \beta/2).
  \end{multline}
  Under the assumption that $\beta$ is even, this can be checked upon the
  manipulation $j \mapsto N - 1 - j$ in the product defining $S_N$, and then use of the
  reflection equation for the gamma functions in that product. Agreement with the
  LHS of (\ref{2.11}) is obtained.
   \end{remark}
   
   We can make use of Corollary \ref{C2.2} and further apply Proposition \ref{P2.1}
   to relate the densities (\ref{2.1}) and (\ref{2.3}), along with the more general $k$-point correlation functions
\begin{multline} \label{n2.16}
\rho_{(k)}^{(Cy)}(x_1,\ldots,x_k):=\frac{N!}{(N-k)!\,\mathcal{N}_N^{(Cy)}}\int_{-\infty}^{\infty}dx_{k+1}\cdots
\\\cdots\int_{-\infty}^{\infty}dx_{N}\,\prod_{l=1}^Nw_{\beta}^{(Cy)}(x_l)\prod_{1\leq j<k\leq N}|x_k-x_j|^{\beta},
\end{multline}
\begin{multline} \label{n2.17}
\rho_{(k)}^{(J)}(x_1,\ldots,x_k):=\frac{N!}{(N-k)!\,\mathcal{N}_N^{(J)}}\int_{-\infty}^{\infty}dx_{k+1}\cdots
\\\cdots\int_{-\infty}^{\infty}dx_{N}\,\prod_{l=1}^Nw_{\beta}^{(J)}(x_l)\prod_{1\leq j<k\leq N}|x_k-x_j|^{\beta}.
\end{multline}
   
  \begin{proposition}\label{P2.4} 
  In the setting of Corollary \ref{C2.2}, the $k$-point correlation functions \eqref{n2.16} and \eqref{n2.17} above are related by
\begin{equation} \label{n2.18}
\rho_{(k)}^{(Cy)}(ix_1,\ldots,ix_k)=(-\cot\,\pi\alpha)^k\rho_{(k)}^{(J)}(x_1,\ldots,x_k)\Big|_{a=b=-\beta(N-1)/2-1-\alpha}.
\end{equation}
In particular, the densities \eqref{2.1} and \eqref{2.3}, themselves being the $1$-point correlation functions, obey the relation
   \begin{equation}\label{2.12}
   \rho_{(1)}^{(Cy)}(ix) = -\cot \pi \alpha \,  \rho_{(1)}^{(J)}(x)  \Big |_{a = b = - \beta (N - 1)/2 - 1 - \alpha}.
  \end{equation}
    \end{proposition}
\begin{proof}
Since $\beta$ is an even integer, inserting the result of Corollary \ref{C2.2} into \eqref{n2.16} shows that the LHS of \eqref{n2.18} is given by
\begin{multline} \label{n2.20}
\frac{N!}{(N-k)!}\frac{(-1)^{\beta N(N-1)/4}(-\cot\,\pi\alpha)^N}{\mathcal{N}_N^{(J)}\big|_{a=b=-\beta(N-1)/2-1-\alpha}}\prod_{l=1}^kw_{\beta}^{(Cy)}(ix_l)\prod_{1\leq h<j\leq k}(ix_j-ix_h)^{\beta}
\\ \times\int_{-\infty}^{\infty}dx_{k+1}(1+x_{k+1}^2)^{\alpha_{\beta,N}}\cdots\int_{-\infty}^{\infty}dx_N(1+x_N^2)^{\alpha_{\beta,N}}
\\ \times\prod_{p=1}^k\prod_{q=k+1}^N(x_q-ix_p)^{\beta}\prod_{k+1\leq h'<j'\leq N}(x_{j'}-x_{h'})^{\beta}.
\end{multline}
From Proposition \ref{P2.1}, interpreting the third line of \eqref{n2.20} as a multivariable symmetric polynomial in $\{ix_l\}_{l=k+1}^N$ shows that the $(N-k)$-fold integral in the above is equal to
\begin{multline*}
\left(-\tan\,\pi\alpha\right)^{N-k}\int_{-1}^1dx_{k+1}(1-x_{k+1}^2)^{\alpha_{\beta,N}}\cdots\int_{-1}^1dx_N(1-x_N^2)^{\alpha_{\beta,N}}
\\\times\prod_{p=1}^k\prod_{q=k+1}^N(ix_q-ix_p)^{\beta}\prod_{k+1\leq h'<j'\leq N}(ix_{j'}-ix_{h'})^{\beta},
\end{multline*}
where we have additionally made the change of variables $x_{k+1}\mapsto-x_{k+1},\ldots,x_N\mapsto-x_N$. Substituting this expression into \eqref{n2.20}, extracting factors of $i$ from the products of differences to cancel the factor of $(-1)^{\beta N(N-1)/4}$, and observing that $w_{\beta}^{(J)}(x)\big|_{a=b=\alpha_{\beta,N}}=w_{\beta}^{(Cy)}(ix)$ gives the RHS of \eqref{n2.18}.
\end{proof}

    \begin{remark} \label{R2.5}
    Although established for $\beta$ even via reasoning based on  Proposition \ref{P2.1},
    an application of Carlson's theorem as familiar in the theory of the Selberg integral
    (see e.g.~\cite[\S 4.1]{Fo10}) shows that Proposition \ref{P2.4} remains true for general $\beta > 0$.
    \end{remark}
  
  For the particular values of $\beta$ even, $\beta = 2, 4$, we will use the identity \eqref{2.12} relating the density of the Cauchy ensemble for $\alpha$ real to (an analytic continuation of) the density of the Jacobi ensemble supported on $(-1,1)$ with $a=b$ and the same value of $\beta$ to study properties of the former. This is presented in Section \ref{S3}, but first we outline a relationship between the Cauchy ensemble when ${\rm Im}\,\alpha\neq0$ (also known as the non-symmetric case or the generalised Cauchy ensemble) and the Jacobi ensemble now requiring $a=\bar{b}$.

   \subsection{The non-symmetric case}
   A classical result of Cauchy \cite{Ca74} gives
  \begin{equation}\label{C.1}
  \int_{-\infty}^\infty {dt \over (1 - i t)^\gamma (1 + i t)^\delta} = 2^{2 - \gamma - \delta}\pi {\Gamma(\gamma+\delta-1) \over \Gamma(\gamma) \Gamma(\delta)}
   \end{equation}
   subject to the requirement that ${\rm Re} \, (\gamma+\delta) > 1$; outside of this range we consider the integral as
   defined by the analytic continuation given by the RHS. Use of the reflection equation for the
   gamma function allows this to be rewritten
  \begin{equation}\label{C.2} 
    \int_{-\infty}^\infty (1 - it)^\gamma (1 + it)^\delta \, dt =    2^{2 + \gamma+\delta} {\sin \pi \gamma \,\sin \pi \delta \over
    \sin \pi (\gamma+\delta) } {\Gamma(\gamma+1) \Gamma(\delta + 1) \over \Gamma (\gamma+\delta + 2)},
   \end{equation}
   subject now to the requirement  ${\rm Re} \, (\gamma+\delta) < - 1$ on the LHS.
   
   The form (\ref{C.2}) is to be compared against the Euler beta function evaluation
   \begin{equation}\label{C.3} 
   \int_0^1 t^c ( 1 - t)^d \, dt = {\Gamma(c+1) \Gamma(d+1) \over \Gamma(c + d + 2)}
     \end{equation}
     or, equivalently,
     \begin{equation}\label{C.3x}
  \int_{-1}^1 (1-t)^c ( 1 + t)^d \, dt =    2^{1+c+d}    {\Gamma(c+1) \Gamma(d+1) \over \Gamma(c + d + 2)} ,  
   \end{equation}    
     where on the LHS, it is required that ${\rm Re} \, c > -1$ and      ${\rm Re} \, d > -1$.
     The agreement in the gamma function dependence of both integrals allows for a relation between
     multiple integrals analogous to that in Proposition \ref{P2.1} to be derived.
   
     \begin{proposition}\label{P2.6} 
     Let $f(x_1,\dots,x_N)$ be a multivariable symmetric polynomial of degree $\tilde{d}$ in each
 $x_i$. For ${\rm Re} \,  (\gamma+\delta) < - \tilde{d} - 1$  define
 \begin{multline}\label{2.7a} 
\tilde{I}_{N,\gamma,\delta}^{(Cy)}[f(x_1,\ldots,x_N)] := \int_{-\infty}^\infty dx_1 \, (1 -  i x_1)^\gamma ( 1 + i x_1)^\delta  \cdots
\\ \cdots \int_{-\infty}^\infty dx_N \,   (1 -  i x_N)^\gamma ( 1 + i x_N)^\delta\,
 f(x_1,\dots,x_N),
 \end{multline}
 and for $\gamma,\delta$ outside of this range, define $ \tilde{I}_{N,\gamma,\delta}^{(Cy)}[f(x_1,\ldots,x_N)] $ by its analytic continuation. Also,
 in relation to the Jacobi weight with $c,d > -1$ define
  \begin{multline}\label{2.8a} 
\tilde{I}_{N,c,d}^{(J)}[f(x_1,\ldots,x_N)] := \int_{-1}^1 dx_1 \, (1-x_1)^c (1 + x_1)^d \cdots
\\ \cdots \int_{-1}^1 dx_N \,  (1 - x_N)^c (1 + x_N)^d \,
 f(x_1,\dots,x_N),
 \end{multline} 
 and for $c,d$ outside of this range, define $ \tilde{I}_{N,c,d}^{(J)}[f(x_1,\ldots,x_N)]$ by its analytic continuation.
 We have
   \begin{equation}\label{2.9a} 
  \tilde{I}_{N,\gamma,\delta}^{(Cy)}[f(1-ix_1,\dots ,1-i x_N )] = \left ( 2 {\sin \pi \gamma\, \sin \pi \delta \over
    \sin \pi (\gamma+\delta) }  \right )^N
     \tilde{I}_{N,\gamma,\delta}^{(J)}[f(1-x_1,\dots  1-x_N)].
  \end{equation} 
  \end{proposition}

  \begin{proof}
  For $p \in \mathbb Z_{\ge 0}$, it follows from (\ref{C.2}) and (\ref{C.3x}) upon setting $c=\gamma$ and $d=\delta$ that
   \begin{equation}\label{C.2a} 
    \int_{-\infty}^\infty (1 - it)^{\gamma+p} (1 + it)^\delta \, dt =    2^{2 + \gamma+\delta+p} {\sin \pi \gamma\, \sin \pi \delta \over
    \sin \pi (\gamma+\delta) } {\Gamma(\gamma+p+1) \Gamma(\delta + 1) \over \Gamma (\gamma+\delta+p + 2)},
   \end{equation}
  and
  \begin{equation}\label{C.3a} 
   \int_{-1}^1 (1 - t)^{\gamma+p} ( 1 + t)^\delta \, dt =  2^{1 + \gamma+\delta+p}  {\Gamma(\gamma+p+1) \Gamma(\delta+1) \over \Gamma(\gamma+\delta+p + 2)}.
     \end{equation}
 Hence, in the sense of analytic continuation,
   \begin{equation}\label{C.4a}     
     \int_{-\infty}^\infty (1 - it)^{\gamma+p} (1 + it)^\delta \, dt =     2 {\sin \pi \gamma\, \sin \pi \delta \over
    \sin \pi (\gamma+\delta) }   \int_{-1}^1 (1-t)^{\gamma+p} ( 1 + t)^\delta \, dt.
  \end{equation}    
     
     The stated result now follows from the assumption that $f(x_1,\ldots,x_N)$ in (\ref{2.7a}) and (\ref{2.8a})
     is a polynomial and so the evaluation of the multiple integrals reduces to the one-dimensional
     integrals (\ref{C.2a}) and (\ref{C.3a}), which are related by (\ref{C.4a}).
     
     \end{proof}
     
We can use Proposition \ref{P2.6} to relate the normalisations (\ref{2.2a}) and (\ref{2.5})
in the case that $\alpha$ in (\ref{2.2}) is complex.
   \begin{corollary}\label{C2.7} 
   Let $\alpha$ be, in general, complex and related to $\alpha_{\beta,N}$ as
   in (\ref{2.2}). For $\beta$ even,
   \begin{equation}\label{2.10a}
(-1)^{\beta N (N - 1) / 4} \mathcal N_N^{(Cy)} = \left( -2{ \sin  \pi \alpha \,  \sin  \pi \bar{\alpha} \over \sin \pi (\alpha + \bar{\alpha}) }  \right)^N  \mathcal N_N^{(J)} \Big |_{\bar{a}  = b = - \beta (N - 1)/2 - 1 - \alpha},
 \end{equation}
 where the RHS is to be regarded as defined by its analytic continuation.

  \end{corollary}    
  
  \begin{proof}
  We observe that for $\beta$ even, the product of differences in the definition of the
  normalisations is a polynomial, and moreover,
  $$ 
  (x_k - x_j)^\beta =  (-1)^{\beta/2}  ((1 - i x_k) - (1 - i x_j))^\beta.
  $$
  The result now follows from the definitions of the normalisations
  and the identity (\ref{2.9a}).
  \end{proof}
  
  \begin{remark} 1.~Analogous to (\ref{2.11}), the identity (\ref{2.10a}) can be
  cast in terms of the integrals in (\ref{2.2a}) and (\ref{2.5}), giving
  \begin{multline}\label{2.11a}
 (-1)^{\beta N (N - 1) / 4} \pi^N M_N(\alpha, \bar{\alpha},\beta/2) \\
 =   \left(-{ \sin  \pi \alpha\,   \sin  \pi \bar{\alpha} \over \sin \pi (\alpha + \bar{\alpha}) }  \right)^N S_N(-\beta(N-1)/2-1-\alpha,
  -\beta (N - 1)/2 - 1 - \bar{\alpha}, \beta/2).
  \end{multline}
 This can be verified using the same steps as for (\ref{2.11}). \\
 2.~In the case $\alpha$ is real, (\ref{2.10a}) reduces to (\ref{2.10}).
 \end{remark}

 We can make use of  Corollary \ref{C2.7} and a further application of
 Proposition \ref{P2.6} in the definitions of the
  $k$-point correlation functions (\ref{n2.16}) and (\ref{n2.17}) to deduce the analogue of Proposition
  \ref{P2.4} in the non-symmetric case. We skip the proof, which essentially amounts to combining the proofs of Proposition \ref{P2.4} and Corollary \ref{C2.7}.
   
  \begin{proposition}\label{P2.9} 
  In the setting of Corollary \ref{C2.7},
   \begin{equation}\label{2.12a}
   \rho_{(k)}^{(Cy)}(ix) = \left(-{\sin \pi (\alpha + \bar{\alpha}) \over 2 \sin \pi \alpha\, \sin \pi \bar{\alpha}}\right)^k  \rho_{(k)}^{(J)}(x)  \Big |_{\bar{a} = b=- \beta (N - 1)/2 - 1 - \alpha}.
  \end{equation}
    \end{proposition}
    
    \begin{remark}
  1.~In the case $\alpha$ is real, (\ref{2.12a}) reduces to (\ref{2.12}).   \\
  2.~For $\beta = 1,2$ and $4$, there are expressions in terms of orthogonal polynomials
  for both sides of (\ref{2.12a}) independently \cite{AFNV00, FLT20}. These can checked
  to be consistent with (\ref{2.12a}), using the fact that the orthogonal polynomials
  associated with the Jacobi weight in  (\ref{1.8}) are the Jacobi polynomials $P_N^{(a,b)}(x)$,
  while those associated with the Cauchy weight in (\ref{1.8}) are the 
  Jacobi polynomials $i^{-N}P_N^{(\eta,\bar{\eta})}(ix)$.
  \end{remark}

  \section{Differential equations}\label{S3}
  \subsection{The symmetric case $\alpha$ real}
  In our previous work \cite{RF19}, a third order linear differential equation was obtained for
  the Jacobi ensemble in the case $\beta = 2$ defined with weight supported on $(0,1)$ specified by
   \begin{equation}\label{3.1}
   \tilde{w}^{(J)}(x)= x^a (1 - x)^b\chi_{0<x<1},
   \end{equation}
  and a fifth order differential equation for this version of the Jacobi ensemble in the case $\beta = 4$ or
  $\beta = 1$. We note that the weight (\ref{3.1}) maps to the Jacobi weight in (\ref{2.4}) by the
  change of variables $x \mapsto (1-x)/2$.  Making this mapping and furthermore setting $b = a$, we
 can read off from \cite[Thms.~2.1 and 2.2]{RF19} the corresponding differential equations
 satisfied by the density $ \rho_{(1)}^{(J)}(x) |_{a=b}=  \rho_{(1)}^{(J)}(x;\beta) |_{a=b}$.
 
 \begin{proposition}\label{P3.1}
 Define
    \begin{multline}\label{3.2} 
   \mathcal D_{2,N}^{(J)} = (1- x^2 )^3 {d^3 \over d x^3} - 8 x ( 1 - x^2 )^2 {d^2 \over d x^2} \\
   - 2 (1- x^2 ) [3 - 2N^2 - 4 a N + (2 (a + N)^2-7) x^2] {d \over dx} \\
   + 4 x ( a^2 + 1  - N^2 - 2 a N + (a + N)^2 x^2 - x^2)
   \end{multline}
   and for $\beta=1$ and $4$,
\begin{multline}\label{3.3} 
   \mathcal D_{\beta,N}^{(J)} = 4(1-x^2)^5{d^5 \over d x^5}-80x(1-x^2)^4{d^4 \over d x^4}+(5\tilde{c}^2-493)(1-x^2)^4{d^3 \over d x^3} \\-4(5\tilde{a}-88)(1-x^2)^3{d^3 \over d x^3}+16(11\tilde{a}-8)x(1-x^2)^2{d^2 \over d x^2}-2(19\tilde{c}^2-539)x(1-x^2)^3{d^2 \over d x^2} \\
+(\tilde{c}^4-64\tilde{c}^2+719)(1-x^2)^3{d \over d x}-8\left[(\tilde{c}^2-45)(\tilde{a}-3)-124\right](1-x^2)^2{d \over d x} \\
+16\left[(\tilde{a}-7)^2-65\right](1-x^2){d \over d x}-(\tilde{c}^2-9)^2x(1-x^2)^2 \\
+4\left[4(\tilde{c}^2-9)+(3\tilde{c}^2-35)\tilde{a}\right]x(1-x^2)-32\tilde{a}^2x,
   \end{multline}
with
\begin{equation*}
\tilde{a}:={a \over \beta/2-1}\left({a \over \beta/2-1}-2\right),\quad \tilde{c}={2a \over \beta/2-1}+4(\beta/2-1)N-1,
\end{equation*}
following \cite{RF19}. Then for $\beta=1,2$ and $4$, we have
\begin{equation} \label{3.4} 
\mathcal D_{\beta,N}^{(J)}\,\rho_{(1)}^{(J)}(x;\beta)|_{a=b}=0.
\end{equation}
   \end{proposition}
   
   \begin{remark}
   Analogous to the differential equations for the densities of classical $\beta$-ensembles considered in
   \cite{RF19}, the differential equations in Proposition \ref{P3.1} have a unique (up to proportionality)
   solution of the form of $w_\beta^{(J)}(x)$ times a polynomial of degree $\beta (N-1)$. The coefficients
   in the latter can be 
   determined by expanding it about infinity, substituting the assumed form in the DE, and equating appropriate
   powers.
   \end{remark}
   
According to Proposition \ref{P2.4}, the differential equation satisfied by $  \rho_{(1)}^{(Cy)}(x)$
for $\beta$ even can be obtained from the one satisfied by $ \rho_{(1)}^{(J)}(x) |_{a=b}$, provided
we set $a =  - \beta (N - 1)/2 - 1 - \alpha$ and replace $x$ by $ix$. Doing this in
Proposition \ref{P3.1} gives us a third order linear differential equation satisfied by
$ \rho_{(1)}^{(Cy)}(x) = \rho_{(1)}^{(Cy)}(x;\beta) $ in the case $\beta = 2$, and a fifth order equation for $\beta = 4$. Furthermore, Remark \ref{R2.5} tells us that the latter differential equation is valid for $\beta=1$, after reparametrising properly. 

\begin{proposition}\label{P3.3}
 Define
    \begin{multline}\label{3.5} 
   \mathcal D_{2,N}^{(Cy)} = (1 + x^2 )^3 {d^3 \over d x^3} + 8 x (1 + x^2 )^2 {d^2 \over d x^2} \\
   + 2 (1 + x^2 ) [3 + 2 N (N + 2\alpha) + (7 -2 \alpha^2) x^2] {d \over dx} 
   + 4 x ( 1 + \alpha^2 + 2 N (N + 2 \alpha) + (1 - \alpha^2) x^2)
   \end{multline}
   and for $\beta=1$ and $4$,
\begin{multline}\label{3.6} 
   \mathcal D_{\beta,N}^{(Cy)} = 4(1+x^2)^5{d^5 \over d x^5}+80x(1+x^2)^4{d^4 \over d x^4}-4(5\tilde{\alpha}-122)(1+x^2)^4{d^3 \over d x^3} \\
+4(5\tilde{N}-93)(1+x^2)^3{d^3 \over d x^3}-8(19\tilde{\alpha}-130)x(1+x^2)^3{d^2 \over d^2 x}+16(11\tilde{N}-19)x(1+x^2)^2{d^2 \over d^2 x} \\
+8(2\tilde{\alpha}^2-31\tilde{\alpha}+82)(1+x^2)^3{d \over d x}-32\left[(\tilde{\alpha}-11)(\tilde{N}-4)-31\right](1+x^2)^2{d \over d x} \\
+16\left[(\tilde{N}-8)^2-65\right](1+x^2){d \over d x}+16(\tilde{\alpha}-2)^2x(1+x^2)^2 \\
-16\left[(3\tilde{\alpha}-8)\tilde{N}+\tilde{\alpha}\right]x(1+x^2)+32(\tilde{N}-1)^2x
   \end{multline}
with
\begin{equation}\label{3.6a}
\tilde{\alpha}:={\alpha\over\beta/2-1}\left({\alpha\over\beta/2-1}-1\right),\qquad \tilde{N}:=\left(2(\beta/2-1)N+{\alpha\over\beta/2-1}\right)^2.
\end{equation}
Then for $\beta=1,2$ and $4$, we have
\begin{equation} \label{3.7} 
\mathcal D_{\beta,N}^{(Cy)}\,\rho_{(1)}^{(Cy)}(x;\beta)=0.
\end{equation}
   \end{proposition}
   
\subsubsection{Relation to Painlev\'e equations in the symmetric case}
   Let $\{ \rho_{(k)}^{(Cy)}\}$ denote the $k$-point correlations for the Cauchy ensemble with
   weight (\ref{2.2}). It is well known (see e.g.~\cite[Ch.~9]{Fo10}) that the generating function
   $E_N^{(Cy)}((s,\infty);\xi)$ for the probabilities $\{ E_N^{(Cy)}(k;(s,\infty)) \}_{k=0}^N$ of there being
   exactly $k$ eigenvalues in the interval $(s,\infty)$ can be written in terms of the correlations
   according to
     \begin{equation}\label{E.1}
     E_N^{(Cy)}((s,\infty);\xi) = 1 + \sum_{k=1}^N { (-\xi)^k \over k!}
     \int_s^\infty dx_1 \cdots  \int_s^\infty dx_k \,  \rho_{(k)}^{(Cy)}(x_1,\dots,x_k).
     \end{equation}
     In the case $\beta = 2$, it is known \cite{WF00} that
  \begin{equation}\label{E.2}     
  \sigma(s) := (1 + s^2) {d \over d s} \log E_N^{(Cy)}((s,\infty);\xi)
  \end{equation}
  satisfies the nonlinear equation (which can be identified in terms of the
  $\sigma$-P${}_{VI}$ equation \cite{FW04}; see also (\ref{PVI}) below)
  \begin{multline}\label{E.3}
  (1+s^2)^2 (\sigma'')^2 + 4 (1 + s^2) (\sigma')^3 - 8 s \sigma (\sigma')^2 \\
  + 4 \sigma^2 (\sigma' - \alpha^2) + 8 \alpha^2 s \sigma \sigma' +
  4 [ N (N + 2 \alpha) - \alpha^2 s^2] (\sigma')^2 = 0.
    \end{multline}  
Note that (\ref{E.3}) is independent of the parameter $\xi$ in (\ref{E.1}).

According to (\ref{E.1}), to leading order in $\xi$,
  \begin{equation}\label{E.4} 
  \sigma(s) = \xi  r(s)  + O(\xi^2), \quad r(s) :=  (1 + s^2)\rho_{(1)}^{(Cy)}(s).
   \end{equation}
Substituting in (\ref{E.3}) and equating terms to the leading order in $\xi$
(which is $O(\xi^2)$) shows
  \begin{equation}\label{E.5} 
 (1 + s^2 )^2 (r''(s))^2 - 4 \alpha^2 ( r(s))^2 + 8 \alpha^2 s r(s) r'(s) +
 4 [ N (N + 2 \alpha) - \alpha^2 s^2] ( r'(s))^2 = 0.
  \end{equation}
  Upon differentiating with respect to $s$, a factor of $r''(s)$ can be
  cancelled and a third order  linear differential equation results,
   \begin{equation}\label{E.6} 
   (1 + s^2)^2 r'''(s) + 2 s (1 + s^2) r''(s) + 4 [ N (N + 2 \alpha) - \alpha^2 s^2] r'(s)
   + 4 \alpha^2 s r(s) = 0.
  \end{equation}
  Recalling the definition of $r(s)$ in terms of $ \rho_{(1)}^{(Cy)}(s)$, we see
  that (\ref{E.6}) is equivalent to the third order differential equation given in Proposition
 \ref{P3.3}.
 
 The use of the characterisation of gap probabilities in terms of $\sigma$-Painlev\'e
 transcendents to derive third order differential equations for the densities of
 other classical ensembles at $\beta = 2$ can be found in \cite{FT18, RF19}.
 
 \begin{remark}
 The $\beta=1,4$ analogue of \eqref{E.6}, obtained by taking $\rho_{(1)}^{(Cy)}(s)=r(s)/(1+s^2)$ in \eqref{3.6}, is
\begin{multline}\label{E.6b}
(1+s^2)^4r^{(5)}(s)+10s(1+s^2)^3r^{(4)}(s)-(5\tilde{\alpha}-22)(1+s^2)^3r'''(s) \\
+(5\tilde{N}-13)(1+s^2)^2r'''(s)-8(\tilde{\alpha}-1)s(1+s^2)^2r''(s) \\
+2(7\tilde{N}+1)s(1+s^2)r''(s)+4\tilde{\alpha}^2(1+s^2)^2r'(s)-2\left[(4\tilde{\alpha}-1)\tilde{N}+1\right](1+s^2)r'(s) \\
+4(\tilde{N}-1)^2r'(s)-4\tilde{\alpha}^2s(1+s^2)r(s)+4\tilde{\alpha}(\tilde{N}-1)sr(s)=0,
\end{multline}
where $\tilde{\alpha}$ and $\tilde{N}$ are as given in Proposition \ref{P3.3}. In comparing \eqref{E.6} and \eqref{E.6b} to
their counterparts in Proposition \ref{P3.3}, we see that the degree of each of the coefficients (which alternate between
being even in $s$, and being odd in $s$) has been reduced by two.
\end{remark}

\subsubsection{The spectrum singularity scaling regime}
It is well known that under the stereographic transformation $s = \cot \theta/2$, the Cauchy ensemble maps to
the circular Jacobi ensemble with weight
\begin{equation}\label{9.2}
|1 - e^{i \theta} |^{2\alpha},\quad\theta\in[0,2\pi)
\end{equation}
and thus eigenvalue PDF proportional to
\begin{equation}
\prod_{l=1}^N|1 - e^{i \theta_l} |^{2\alpha} \prod_{1 \le j < k \le N} |e^{i \theta_k} - e^{i \theta_j}|^\beta;
\end{equation}
see e.g.~\cite[\S 2.5]{Fo10}, \cite{Li17}. In particular, the densities in the two
ensembles are related by 
\begin{equation}\label{9.3}
\rho_{(1)}^{(Cy)}(s) = 2\sin^2 {\theta\over2}\,\rho_{(1)}^{(cJ)}(\theta).
\end{equation}
Equivalently, in the notation $r(s)$ of (\ref{E.4}),
\begin{equation}\label{9.3a}
r(s) = 2  \rho_{(1)}^{(cJ)}(\theta).
\end{equation}
In the case $\alpha=0$, the weight (\ref{9.2}) is a constant and so according to (\ref{9.3a}), $r(s)$
is then also a constant. We can see immediately that this is consistent with equations (\ref{E.6}) and (\ref{E.6b}).

Suppose now $\theta$ is scaled by writing $\theta = 2 X/N$. The circular Jacobi ensemble then admits a
well defined scaling limit as $N \to \infty$, giving rise to what is termed a spectrum singularity at $\theta=0$
\cite[\S 3.9]{Fo10}. In view of (\ref{9.3a}), in the cases $\beta = 1,2$ and $4$, differential equations for the corresponding density
$\rho_{(1)}^{(s.s.)}(X)$ can be obtained by setting 
\begin{equation}\label{9.3x}
s = N/X
\end{equation}
 and $r(s) = \rho_{(1)}^{(s.s.)}(X)$
 in (\ref{E.6}) and (\ref{E.6b}), and then equating terms at leading order in $N$.
 Specifically for $\beta = 2$, we therefore have that $R(X) = \rho_{(1)}^{(s.s.)}(X)$
satisfies the third order linear differential equation
\begin{equation}\label{9.3b}
X^2  R'''(X) + 4 X R''(X) + (2 - 4 \alpha^2 + 4 X^2) R'(X) - {4 \alpha^2 \over X} R(X) = 0.
\end{equation}
This is consistent with the known exact formula \cite{NS93}, \cite[Eq.(7.49) with $\pi \rho = 1$]{Fo10}
\begin{equation}\label{9.3c}
 \rho_{(1)}^{(s.s.)}(x) = {x \over 2} \Big ( ( J_{\alpha - 1/2}(x) )^2 + ( J_{\alpha + 1/2}(x) )^2 - {2 \alpha \over x}
 J_{\alpha - 1/2}(x)   J_{\alpha + 1/2}(x)  \Big ),
 \end{equation}
 as can be checked using computer algebra.

\subsubsection{The global scaling of the symmetric Cauchy density} \label{s3.1.3}
The scaling (\ref{9.3x}) corresponds to a spacing of order unity for eigenvalues in the neighbourhood of the
 spectrum singularity. This is in contrast to the global scaling of the
 density,
\begin{equation} \label{3.23}
\rho_{(1)}^{(glob.\,Cy)}(s):=\lim_{N \to \infty} {1 \over N}\rho_{(1)}^{(Cy)}(s) \Big |_{\alpha = \hat{\alpha} N},\quad\hat{\alpha}\textrm{ constant in }N,
\end{equation}
where the spacing between eigenvalues is of order $1/N$. Upon replacing $\alpha$ by $\hat{\alpha}N$ in (\ref{E.6}), equating the leading order
 terms in $N$ gives
 \begin{equation}\label{9.3y}
 ((1+2 \hat{\alpha}) - \hat{\alpha}^2 s^2) r'(s) +  \hat{\alpha}^2 s r(s) = 0.
  \end{equation}
  Solving this first order equation, and making use of the definition of $r(s)$ from (\ref{E.4}), we conclude for $\beta=2$ that
\begin{equation}\label{9.3z}  
\rho_{(1)}^{(glob.\,Cy)}(s) = 
A { \sqrt{1 + 2 \hat{\alpha} - \hat{\alpha}^2 s^2} \over 1 + s^2} 
\chi_{|s| < \sqrt{1 + 2 \hat{\alpha}}/ \hat{\alpha}}, \quad A = {1 \over \pi },
  \end{equation}
  where the scalar $A$ has been determined by the requirement that the RHS integrate to unity. Equation \eqref{9.3z} is expected to hold universally for $\beta>0$, as the limiting densities of the other classical ensembles are known to be independent of $\beta$. This is in fact known from the equivalence \eqref{9.3} of the Cauchy ensemble to the circular Jacobi ensemble upon the stereographic transformation $s=\cot\theta/2$. Thus, it has been established \cite{BNR09,JJK10} that for all $\beta>0$, the limiting equilibrium density for the circular Jacobi ensemble with $\alpha=\hat{\alpha}\beta N/2$ is given by the stereographic transformation of \eqref{9.3z},
\begin{equation} \label{3.25a}
\rho_{(1)}^{(glob.\,cJ)}(\theta)={\hat{\alpha}\over 2\pi}\sqrt{\cot^2{\theta_c\over 2}-\cot^2{\theta\over2}}\,\chi_{\theta_c<\theta<2\pi-\theta_c}
\end{equation}
where
\begin{equation*}
\theta_c=\min_{\theta\in[0,2\pi)}\{\theta\,:\,\cot^2\theta/2=(1+2\hat{\alpha})/\hat{\alpha}^2\}.
\end{equation*}
 
   \subsection{The non-symmetric case Im$\, \alpha \ne 0$}
   In the symmetric case, differential equations for $\beta = 1,2$ and
   $4$ were considered. This is again possible in the non-symmetric case; however,
   the complexity of the case $\beta = 4$ (equivalently $\beta=1$) increases to the extent that it becomes cumbersome
   to present (for the Jacobi ensemble defined on $(0,1)$ its explicit form is given in
   \cite[Th.~2.2]{RF19}), so we will restrict attention to the case $\beta = 2$.
   The first task is to change variables $x \mapsto (1 - x)/2$
   in the known third order linear differential equation
   for the density in the Jacobi ensemble on $(0,1)$ with $\beta = 2$ \cite[Th.~2.1]{RF19}.
   Unlike the symmetric case considered above, we no longer set $a=b$, so this does not change the differential equation of \cite{RF19} in any
   essential way, and thus the result of this calculation will not be recorded here. However, its
   explicit form is required in what comes next.  Specifically, starting from the explicit form and
   setting $a = - N - \bar{\alpha}$, $b = - N - \alpha$, and replacing $x$
   by $ix$, it follows from (\ref{2.12a}) that we obtain the differential equation satisfied by
   $\rho_{(1)}^{(Cy)}(x)$.
   
 \begin{proposition}\label{P3.6}  
 Define the third order differential operator
 \begin{multline}\label{3.15}
 \tilde{\mathcal{D}}_{2,N}^{(Cy)} = (1 + x^2)^3 {d^3 \over d x^3} + 8 x ( 1 + x^2)^2 {d^2 \over d x^2} \\
 + (1 + x^2)[6 + 4N (N + \alpha + \bar{\alpha}) + (\alpha-\bar{\alpha})^2  + 2 i (\alpha-\bar{\alpha}) (2 N + 
 \alpha + \bar{\alpha}) x + (14 - (\alpha + \bar{\alpha})^2) x^2 ] {d \over d x} \\
 + [4 + 8 N(N+\alpha+\bar{\alpha}) + 3 \alpha^2 - 2 \alpha \bar{\alpha} + 3 \bar{\alpha}^2 +(4-(\alpha+\bar{\alpha})^2)x^2]x \\
+ i (\alpha-\bar{\alpha}) ( 2N + \alpha + \bar{\alpha} )(3x^2-1).
 \end{multline}
 For $\beta = 2$, we have
  \begin{equation}\label{E.6a} 
   \tilde{\mathcal{D}}_{2,N}^{(Cy)} \, \rho_{(1)}^{(Cy)}(x) = 0.
     \end{equation}
 \end{proposition}
 
\subsubsection{Relation to Painlev\'e equations in the non-symmetric case}
After a simple change of variables, the Jimbo-Miwa-Okamoto $\sigma$-form of the
 Painlev\'e differential equation reads \cite[Eq.~(1.32)]{FW04}
  \begin{equation}\label{PVI}
  h'\Big ( (1 + t^2)h'' \Big )^2 + 4 \Big ( h' ( h - t h') - i
  b_1b_2b_3b_4 \Big )^2 + 4 \prod_{k=1}^4 (h' + b_k^2) = 0.
 \end{equation}
 Let $\mathbf b = (b_1, b_2, b_3, b_4)$ and define $e_2'[\mathbf b]$, $e_2[\mathbf b]$ as the
 elementary symmetric polynomials of degree two in $\{b_1, b_3, b_4 \}$ and
 $\{b_1, b_2,b_3, b_4 \}$, respectively.
 Set 
   \begin{equation}\label{PVI.1}
   U_N^{(Cy)}(t;(\alpha_1,\alpha_2);\xi) =
   (t^2 + 1) {d \over d t} \log \Big (
   (it - 1)^{e_2'[\mathbf b] - e_2[\mathbf b]/2} (it + 1)^{e_2[\mathbf b]/2}
   E_N^{(Cy)}((t,\infty);\xi)) \Big ),
 \end{equation}
 where $E_N^{(Cy)}$ is specified by (\ref{E.1}) in the non-symmetric case with
 $\alpha = \alpha_1 + i \alpha_2$ ($\alpha_2 \ne 0$). We have from \cite[Prop.~15]{FW04} that
 $U_N^{(Cy)}$ satisfies the transformed $\sigma$-P${}_{VI}$ equation (\ref{PVI}) with parameters
  \begin{equation}\label{PVI.2} 
 \mathbf b = ( - \alpha_1, -i \alpha_2, N + \alpha_1, \alpha_1).
  \end{equation}
 Analogous to (\ref{E.4}), we have from (\ref{E.1}), (\ref{PVI.1}) and (\ref{PVI.2}) that
 \begin{equation}\label{PVI.3}  
  U_N^{(Cy)}(t;(\alpha_1,\alpha_2);\xi) = (N + \alpha_1) \alpha_2 - \alpha_1^2 t + \xi r(t) + O(\xi^2), \qquad
  r(t) := (1 + t^2)   \rho_{(1)}^{(Cy)}(t).
  \end{equation}
  
  Substituting (\ref{PVI.3}) in (\ref{PVI}) and equating terms of leading order in $\xi$, which
  as for   (\ref{E.4})  occurs at order $\xi^2$, then differentiating and cancelling a factor
  of $r''$ shows
  \begin{multline}\label{3.27}
   (1 + t^2)^2 r''' +  2   t (1 + t^2) r''  + 4[ ( \alpha_1^2 t+(N + \alpha_1) \alpha_2 ) (r - t r') \\
 +  ((N + \alpha_1)^2 -  (N + \alpha_1) \alpha_2 t - \alpha_1^2 - \alpha_2^2 ) r']  = 0.
  \end{multline}
  In the case $\alpha_2 = 0$, this agrees with (\ref{E.6}).
  Now, substituting for $r(t)$ in terms of $ \rho_{(1)}^{(Cy)}(t)$ as specified in (\ref{PVI.3}) reclaims
  (\ref{3.15}) and (\ref{E.6a}).

\subsubsection{The global scaling of the non-symmetric Cauchy density}
The global scaling of the density $\rho_{(1)}^{(glob.\,Cy)}(s)$ specified by \eqref{3.23} with $\hat{\alpha}=\hat{\alpha}_1+i\hat{\alpha}_2$, generalising the working of \S\ref{s3.1.3}, is deduced from
  (\ref{3.27}) by setting $\alpha_1 = \hat{\alpha}_1N$, $\alpha_2 =\hat{\alpha}_2N$, equating terms of
  leading order in $N$, solving the resulting first order differential equation, and finally substituting
  for $r(t)$ in terms of $ \rho_{(1)}^{(Cy)}(t)$ as specified in (\ref{PVI.3}). The final result is
\begin{equation}\label{9.3zv}  
 \rho_{(1)}^{(glob.\,Cy)}(s) = 
A { \sqrt{(u_+-s)(s-u_-)} \over 1 + s^2} 
\chi_{u_- < s < u_+} ,
  \end{equation}
  where
\begin{equation}\label{9.3zw}   
u_\pm = {-(1 +  \hat{\alpha}_1)  \hat{\alpha}_2 \pm \sqrt{( \hat{\alpha}_1^2+ \hat{\alpha}_2^2)(1 + 2  \hat{\alpha}_1)} \over  \hat{\alpha}_1^2}, \quad  A = {\hat{\alpha}_1 \over \pi}.
  \end{equation}
  As expected, setting $\hat{\alpha}_2=0$ in \eqref{9.3zv} reclaims \eqref{9.3z}. Also, as with \eqref{9.3z}, the transformation of \eqref{9.3zv} to the unit circle is known in the context of the study of the circular Jacobi ensemble \cite[Eq.~(5.6)]{BNR09}.
  
  \section{Moments}\label{S4}
  \subsection{Moment recurrences for the symmetric cases}
  Restricting attention at first to non-negative integer moments, in the symmetric Jacobi ensemble, only the even moments are non-zero.
  The change of variables $x \mapsto 1 - 2 x $ shows that in this setting,
   \begin{equation}\label{S.1}
   \Big \langle \sum_{l=1}^N x_l^{2k} \Big \rangle_{(-1,1)}^{(J)} =
   \Big \langle \sum_{l=1}^N (1 - 2 x_l )^{2k} \Big \rangle_{(0,1)}^{(J)} =
   \sum_{s=0}^{2k} \binom{2k}{s}  (-2)^{s}      \Big \langle \sum_{l=1}^N x_l^{s} \Big \rangle_{(0,1)}^{(J)} ,
   \end{equation}
   where in the second and third expression, the Jacobi ensemble is specified by the weight
   (\ref{3.1}) supported on $(0,1)$ (as indicated by the notation) and with $a=b$; the second
   equality follows from the binomial expansion. The averages on the RHS for $s=0$ and
   $s=1$ are immediate,
    \begin{equation}\label{T.2}  
     \Big \langle \sum_{l=1}^N x_l^s \Big \rangle_{(0,1)}^{(J)} \Big |_{s=0} = N, \qquad
      \Big \langle \sum_{l=1}^N x_l^s \Big \rangle_{(0,1)}^{(J)} \Big |_{s=1} = {N \over 2},
     \end{equation}   
       as follows from the normalisation and symmetry of the distribution about $x_l=1/2$, respectively.
   For $s \ge 2$
   the averages on the RHS of (\ref{S.1})
   are known from Jack polynomial theory; see \cite[Eq.~(53)]{MRW15} and
   \cite[Eqns.~(33)-(35)]{FLD16}. In particular, we read off from Eq.~(189) of \cite{FLD16} (must multiply this equation by $N$ and replace $\beta^2$ by $-\beta/2$ for the present setting) and Eq.~(B.7b) of \cite{MRW15} that
   \begin{multline}\label{S.2}  
    \Big \langle \sum_{l=1}^N x_l^{2} \Big \rangle_{(0,1)}^{(J)}  = N \\
    \times {2 a^2 + a [- (\beta/2) ( 6 - 5 N) + 6] + (\beta/2)^2 (N-1)(3N-4)  -(\beta/2) (9-7N) + 4 \over
    2 (2a - (\beta/2) (3 - 2N) + 2) ( 2 a +  \beta (N - 1) + 3 ) }.
    \end{multline}

    For integer values of $k$, the Jack polynomial formula leading to (\ref{S.2}) reveals that 
    the $2k$-th moment is a rational function in $N$. The recent work \cite{CMOS18} has drawn
    attention to analogous features of the moments as a function of $k$ in the complex plane
    for the Gaussian,  Laguerre and non-symmetric Jacobi (the
    latter defined on $(0,1)$)
    classical ensembles. Moreover, for $\beta = 2$, it was found that this function of $k$
    can be factored in terms of some gamma functions times a particular
    hypergeometric polynomial of degree $N-1$ from the Askey table. In a subsequent work
     \cite{ABGS20}, these considerations were extended to the symmetric Cauchy ensemble with
     $\beta = 2$, which involved the continuous Hahn polynomials.
     Taking the viewpoint of Section \ref{S2} that integrations over the Cauchy
     ensemble can be obtained as corollaries of integrations over the Jacobi ensemble defined on
     $(-1,1)$, our aim in this section is to give a self-contained derivation of results of this type for the  symmetric
     Jacobi ensemble with $\beta = 2$. The essential idea is that the differential equations of Section \ref{S3} lead to recurrences for the moments.
     
     First, note from (\ref{1.9}) that for a general $\beta$-ensemble with $\beta$ even, we have
      \begin{equation}\label{S.2x}   
   \rho_{(1)}(x) = w_\beta(x) q_{\beta ( N - 1)/2}(x^2)
 \end{equation}
 for some polynomial $ q_{\beta ( N - 1)/2}(y)$ of degree $\beta (N - 1)/2$. 
 Specialising now to the symmetric Jacobi weight (\ref{2.4}) with $a=b$, and
 making use of the integration formula (\ref{6.2b}), it then follows that
   \begin{equation}\label{S.2y}
   \int_{-1}^{1} |x|^{2k}   \rho_{(1)}^{(J)}(x) \, dx = {\Gamma(k+1/2) \over \Gamma(k + \beta(N-1)/2+a + 3/2)} Q_{\beta (N - 1)/2}(k)
   \end{equation} 
   for some polynomial $Q_{\beta (N - 1)/2}(y)$ of degree $\beta (N - 1)/2$. 
    
 It turns out that the recurrences for the symmetric Jacobi ensemble's moments simplify if written in terms of the differences of successive even moments
     \begin{equation}\label{S.2a} 
 \mu_k^{(J)} = m_{2k+2}^{(J)} - m_{2k}^{(J)}, \qquad m_{2k}^{(J)}:=   \int_{-1}^1 |x|^{2k} \rho_{(1)}^{(J)}(x) |_{a=b} \, dx.
  \end{equation}
 Note that as a function of $k$, this is well defined for ${\rm Re} \, k > - 1/2$. The equivalent quantities in the symmetric Cauchy case are the sums of successive even moments
\begin{equation}\label{4.7}
\mu_k^{(Cy)}=m_{2k+2}^{(Cy)}+m_{2k}^{(Cy)}, \qquad m_{2k}^{(Cy)}:=\int_{-\infty}^{\infty}|x|^{2k}\rho_{(1)}^{(Cy)}(x)\,dx.
\end{equation}
For $\mu_k^{(Cy)}$ to be well defined, we require $-1/2 < {\rm Re} \, k < \alpha - 1/2$, the upper bound now necessary due to the domain of integration being non-compact. The respective differences and sums are equivalent in the sense that, via analytic continuation,
\begin{equation} \label{4.8}
\mu_k^{(J)}\Big |_{a \mapsto - \beta(N-1)/2 - 1 - \alpha} = (-1)^{k-1} \mu_k^{(Cy)},
\end{equation}
as is consistent with Proposition \ref{P2.1}.
    
 In relation to the difference (\ref{S.2a}), the natural quantity is
 \begin{equation}\label{S.3}   
 r^{(J)}(x) = (1 - x^2)    \rho_{(1)}^{(J)}(x) |_{a=b}.
  \end{equation}
   It follows from Proposition \ref{P3.1} that for $\beta = 2$, $ r^{(J)}$ satisfies
   the differential equation
 \begin{equation}\label{E.6x} 
   (1 - x^2)^2 r'''(x) - 2 x (1 - x^2) r''(x) + 4 [ N (N + 2 a) - (N + a)^2 x^2] r'(x)
   + 4  (a + N)^2 x r(x) = 0,
  \end{equation}
while for $\beta=1$ and $4$, $r^{(J)}$ instead satisfies
\begin{multline} \label{4.11}
4(1-x^2)^4r^{(5)}(x)-40x(1-x^2)^3r^{(4)}+(5\tilde{c}^2-93)(1-x^2)^3r'''(x)
\\-4(5\tilde{a}-8)(1-x^2)^2r'''(x)-8(\tilde{c}^2+5)x(1-x^2)^2r''(x)+8(7\tilde{a}+18-10x^2)x(1-x^2)r''(x)
\\+(\tilde{c}^2-1)^2(1-x^2)^2r'(x)-8\left[\tilde{c}^2(\tilde{a}+1)-2\tilde{a}-1\right](1-x^2)r'(x)
\\+16\tilde{a}^2r'(x)+(\tilde{c}^2-1)\left[(\tilde{c}^2-1)(1-x^2)-4\tilde{a}\right]xr(x)=0,
\end{multline}
with $\tilde{a}$ and $\tilde{c}$ as given in Proposition \ref{P3.1}. These differential equations can formally be obtained from (\ref{E.6}) and \eqref{E.6b} by the mappings $x \mapsto i x$,
  $\alpha \mapsto - \beta(N-1)/2-1 - a$, or equivalently, one can simply take $r^{(J)}(x)=\rho_{(1)}^{(J)}(x)/(1-x^2)$ in equations \eqref{3.2}--\eqref{3.4}.
  Note the simplification relative to (\ref{3.2}) and \eqref{3.3}.
  From these differential equations, second ($\beta=2$) and fourth ($\beta=1,4$) order recurrences can be
  derived for $\{ \mu_{k+j}^{(J)} \}_{j \in \mathbb Z}$. The method,
  which is based on integration by parts, requires that the symmetric Jacobi parameter
  $a$ be greater than zero, although this condition is not necessary, and the recurrence
  correctly specifies the sequence for all values that it is well defined.

 \begin{proposition} \label{P4.1}
 Define $ \mu_k^{(J)}$ by (\ref{S.2a}) and recall the definitions of $\tilde{a}$ and $\tilde{c}$ given in Proposition \ref{P3.1}. For $\beta = 2$, the sequence  $\{ \mu_{k+j}^{(J)} \}_{j \in \mathbb Z}$
 satisfies the recurrence
  \begin{multline}\label{E.7a} 
 (2 k + 4)[  (2 k + 3)^2 - 4 (a + N)^2  ] \mu_{k+1} - 2  (2 k + 1) 
 [  ( 2k+2)^2  - 2 N ( N + 2a)] \mu_k \\
 + (2k + 1 ) (2k) (2 k - 1) \mu_{k-1} = 0.
 \end{multline}
The initial condition determining the sequence $\{ \mu_k^{(J)} |_{\beta = 2} \}_{k=0}^\infty$ is
\begin{equation} \label{4.13}
\mu_0^{(J)} = {2 N (a + N) (2 a + N)  \over 1 - 4 (a + N)^2}.
\end{equation}
For $\beta=1$ and $4$, the sequence $\{ \mu_{k+j}^{(J)} \}_{j \in \mathbb Z}$
 satisfies the recurrence
\begin{equation} \label{4.14}
\sum_{l=-2}^2f_l\,\mu_{k+l}^{(J)}=0,
\end{equation}
where
\begin{align*}
f_{-2}&:=-4(2k+1)(2k)(2k-1)(2k-2)(2k-3),\\
f_{-1}&:=(2k+1)(2k)(2k-1)\left[20\tilde{a}-5\tilde{c}^2+16k(4k+5)+77\right],\\
f_0&:=(2k+1)\Big[ 8\tilde{a}(5k+8)(2k+3)-\left(\tilde{c}^2-4\tilde{a}-30k^2-67k-42\right)^2 \\
&\qquad+516k^4+2292k^3+3653k^2+2534k+673\Big], \displaybreak
\\f_1&:=\tilde{c}^2\left[(\tilde{c}^2-4\tilde{a})(4k+7)-120k^3-656k^2-1210k-746\right] \\
&\quad+\tilde{a}\left[160k^3+736k^2+1128k+580\right] \\
&\quad+512k^5+4480k^4+16056k^3+29360k^2+27270k+10243, \\
f_2&:=-2(k+3)(\tilde{c}+4k+11)(\tilde{c}+2k+4)(\tilde{c}-2k-4)(\tilde{c}-4k-11).
\end{align*}
The initial conditions determining the sequences $\{\mu_k^{(J)} |_{\beta =1,4} \}_{k=0}^{\infty}$ are
\begin{align}
\mu_0^{(J)}&=\frac{(\tilde{c}-1)\left(\tilde{c}+\tfrac{2a}{\beta/2-1}-3\right)\left(\tilde{c}-\tfrac{2a}{\beta/2-1}+1\right)}{8\tilde{c}(\tilde{c}-3)(1-\beta/2)},  \label{4.15}
\\ \mu_1^{(J)}&=\frac{(\tilde{c}^2-5)(\tilde{c}-7)-4\tilde{a}(\tilde{c}-1)}{4(\tilde{c}^2-4)(\tilde{c}-7)}\mu_0^{(J)}. \label{4.16}
\end{align}
 \end{proposition}
 
 \begin{proof}
 We first give the details for the $\beta=2$ recurrence \eqref{E.7a}. Begin by multiplying (\ref{E.6a}) by $x |x|^{2k}$ and then integrate over $x$ from $-1$ to $1$.
 From the definition (\ref{S.2a}), the final term of the differential equation transforms
 to
 $
 - 4 (a + N)^2 \mu_{k+1}.
 $
 For the other terms, integration by parts is required to reduce the integrations to the form
 of (\ref{S.2a}). Due to the absolute value sign in the multiplying term $x |x|^{2k}$, this
 requires considering the intervals $(0,1)$ and $(-1,0)$ separately. 
 The endpoints do not contribute to the integration.
 At  the origin this is due to the
 term $x |x|^{2k}$  with $k > - 1/2$, and at $\pm 1$, there is no contribution
   since
$ \rho_{(1)}^{(J)}(x) |_{a=b}$ is then zero for $a > 0$. 
After simplification, (\ref{E.7a}) results.

To obtain recurrence \eqref{4.14}, apply these same steps to equation \eqref{4.11}.
For both recurrences, the value of $\mu_0^{(J)}$ follows from \eqref{S.1}--\eqref{S.2}, while $\mu_1^{(J)}$ can be computed using MOPS \cite{MOPS} and equation \eqref{S.1}.
\end{proof}

\begin{remark}
1.~The analogue of (\ref{E.6x}) for the $\beta=2$ non-symmetric Jacobi case can be obtained from 
(\ref{3.27}) by the mappings $t \mapsto -i x$, $\alpha_1 \mapsto - N - (a+b)/2$ and $\alpha_2\mapsto (b-a)i/2$. Repeating the working which gave (\ref{E.7a}) leads to
  a recurrence which also involves $\mu_{k+1/2}$. In fact, this recurrence is given in \cite[Eq.~(31)]{Le04}, which reduces to \eqref{E.7a} upon setting $a=b$. For $\beta=1$ and $4$, one may take existing moment recurrences for the non-symmetric Jacobi ensemble on $(0,1)$ \cite{RF19} and apply equation \eqref{S.1}.\\
  2.~Scaling $x \mapsto x/ \sqrt{2a}$ in the definition (\ref{S.2a}) of $m_{2k}^{(J)}$ shows
  $$
  m_{2k}^{(J)} = {1 \over (2a)^{k+1/2}} \int_{-\sqrt{a}}^{\sqrt{a}} |x|^{2k} \rho_{(1)}^{(J)}(x/\sqrt{2a}) \Big |_{a=b} \, dx.
  $$
  Now, from the elementary limit $(1 - x^2/ 2a)^a \to e^{-x^2/2}$ as $a \to \infty$, it follows from the
  definition of $\mu_k^{(J)}$ that
  $$
  \lim_{a \to \infty} (2a)^{k+1/2}  \mu_k^{(J)}  = - m_{2k}^{{\rm GUE}^*},
  $$
  where GUE${}^*$ refers to the ensemble (\ref{1.9}) with weight $w(x) = e^{- x^2/2}$ and $\beta = 2$.
  Multiplying (\ref{E.7a}) by $(2 a)^{k-1/2}$ and taking $a \to \infty$ thus implies the
  recurrence for $\{ m_{2k}^{{\rm GUE}^*} \}$,
   \begin{equation}\label{E.6y} 
   - (k+2)  m_{2k+2}^{{\rm GUE}^*} + 2 N (2 k + 1)  m_{2k}^{{\rm GUE}^*}  + k (2k + 1) (2k - 1)
 m_{2k-2}^{{\rm GUE}^*}   = 0,
 \end{equation}
 found originally by Harer and Zagier \cite{HZ86}. Likewise, multiplying \eqref{4.14} by $(2a)^{k-3/2}$ and taking $a\to\infty$ recovers the recurrences for the GOE and GSE moments given in \cite[Thm.~2]{Le09}, \cite[Thms.~11 and 17]{WF14}.
  \end{remark}

As a simple consequence of the relation \eqref{4.8}, Proposition \ref{P4.1} leads to recurrences on the sums of successive even moments $\mu_k^{(Cy)}$ of the $\beta=1,2$ and $4$ Cauchy ensembles.

\begin{corollary} \label{C4.3}
Define $\mu_k^{(Cy)}$ by \eqref{4.7} and retain the definitions of $f_{-2},\ldots,f_2$ given in Proposition \ref{P4.1}. For $\beta=2$, the sequence $\{\mu_{k+j}^{(Cy)}\}_{j\in\mathbb{Z}}$ satisfies the recurrence
\begin{multline}\label{4.21}
(2k+4)\left[(2k+3)^2-4\alpha^2\right]\mu_{k+1}^{(Cy)}+2(2k+1)\left[(2k+2)^2+2N(N+2\alpha)\right]\mu_k^{(Cy)} \\
+(2k+1)(2k)(2k-1)\mu_{k-1}^{(Cy)}=0.
\end{multline}
The initial condition determining the sequence $\{ \mu_k^{(Cy)} |_{\beta = 2} \}_{k=0}^\infty$ is
\begin{equation} \label{4.21a}
\mu_0^{(Cy)} = {2N \alpha (N + 2 \alpha) \over (2 \alpha - 1) (2 \alpha + 1)},\quad\alpha>1/2.
\end{equation}
For $\beta=1$ and $4$, the sequence $\{\mu_{k+j}^{(Cy)}\}_{j\in\mathbb{Z}}$ satisfies the recurrence
\begin{equation} \label{4.20}
\sum_{l=-2}^2g_l\,\mu_{k+l}^{(Cy)}=0,
\end{equation}
where $g_l:=(-1)^{l-1}f_l\big|_{a\mapsto -\beta(N-1)/2-1-\alpha}$ with initial conditions given by parsing equations \eqref{4.15} and \eqref{4.16} through \eqref{4.8}.
\end{corollary}

\begin{remark}
Let $\tilde{m}_k^{(J)} = \tilde{m}_k^{(J)}(N,\beta,a,b)$ denote the $k$-th moment of the spectral
density for the Jacobi $\beta$-ensemble defined on $(0,1)$. It is known
\cite{DP12,FLD16,FRW17} that
   \begin{equation}\label{4.30x}
\tilde{m}_k^{(J)}(N,\beta,a,b) = - (2/\beta) \tilde{m}_k^{(J)}(-\beta N/2,4/\beta,-2a/\beta,-2b/\beta).
\end{equation}
It follows from (\ref{S.1}) that this remains true for the moments $m_k^{(J)}$ of the spectral
density for the Jacobi $\beta$-ensemble defined on $(-1,1)$. Suppose now we set
$$
a = -\beta (N - 1)/2 - 1 - \bar{\alpha}, \qquad b = - \beta (N - 1)/2 - 1 - \alpha.
$$
Then we know from Proposition \ref{P2.6} that $m_k^{(J)}$ equals $(-1)^{k/2} m_k^{(Cy)}(N,\beta,\alpha,\bar{\alpha})$. Relating
the RHS of (\ref{4.30x}) to $m_k^{(Cy)}$ for renormalised parameters shows
  \begin{equation}\label{4.31x}
  m_k^{(Cy)}(N,\beta,\alpha,\bar{\alpha}) =  - (2/\beta)  m_k^{(Cy)}(-\beta N/2, 4/\beta,-2\alpha/\beta,-2 \bar{\alpha}/\beta).
  \end{equation}
  It can be checked that recurrence \eqref{4.20} respects the duality \eqref{4.31x}.
  \end{remark}

\subsection{Continuous Hahn polynomials}
The $\beta=2$ recurrence \eqref{E.7a} simplifies upon introducing the rescaling
  \begin{equation}\label{S.2b} 
  \mu_{k}^{(J)}  =  {\Gamma(k+1/2) \over \Gamma(k +  \beta (N-1)/2  + a + 5/2) }  \tilde{\mu}_k^{(J)},
 \end{equation}
as motivated by (\ref{S.2y}). According to this rescaling, we have
 \begin{align*}
 \mu_{k+1}^{(J)} & = - {k + 1/2 \over N + a + k + 3/2} {\Gamma(k+1/2) \over  \Gamma(N + a +  k + 3/2)} \tilde{\mu}_{k+1}^{(J)}, \\
 \mu_{k-1}^{(J)} & = - {N + a + k + 1/2 \over k - 1/2}  {\Gamma(k+1/2)  \over \Gamma(N + a + k + 3/2)}  \tilde{\mu}_{k-1}^{(J)}  .
 \end{align*}
 Substituting these expressions into (\ref{E.7a}) gives a three term recurrence with coefficients that are quadratic rather
 than cubic in $k$. (We remark that this simplification does not extend to the $\beta=1,4$ recurrence \eqref{4.14}.)
 
 \begin{corollary}
 Define $ \tilde{\mu}_k^{(J)}$ by (\ref{S.2b}) and (\ref{S.2a}). For $\beta = 2$, the sequence  $\{ \tilde{\mu}_{k+j}^{(J)} \}_{j \in \mathbb Z}$
 satisfies the recurrence
  \begin{multline}\label{E.7b} 
 (2 k + 4)[  (2 k + 3) - 2 (a + N)  ] \tilde{\mu}_{k+1} + 2  
 [  ( 2k+2)^2  - 2 N ( N + 2a)] \tilde{\mu}_k \\
 + (2k ) (2 (N + a + k) + 1)  \tilde{\mu}_{k-1} = 0.
 \end{multline}
 \end{corollary}
 
 \begin{remark}\label{R4.4}
 For a given value of $\tilde{\mu}_0$, the recurrence (\ref{E.7b}) uniquely
 specifies $\{ \tilde{\mu}_k^{(J)} \}_{k=1}^\infty $. With $\tilde{\mu}_0$ defined by
 (\ref{4.13}) and (\ref{S.2b}), we have that 
  \begin{equation}\label{E.7c}
 \tilde{\mu}_0^{(J)} =  {\Gamma(N + a + 3/2) \over \Gamma(1/2)}
 \Big ( {2 N (a + N) (2 a + N)  \over 1 - 4 (a + N)^2}  \Big ).
 \end{equation}
 \end{remark}
 
 At this stage, following \cite{ABGS20}, introduce the continuous Hahn polynomials
 \begin{equation}\label{H}
 S_n(x;a,b,c,d) = i^n {(a+c)_n (a + d)_n \over n!} \,
 \hypergeometric{3}{2}{-n,n+a+b+c+d-1,a+ix}{a+c,a+d}{1},  
 \end{equation}
  where on the RHS the notation is standard as for  hypergeometric functions.
  Further specialise these polynomials by writing
  $s_n(x;a,b) = S_n(x;a,b,\bar{a},\bar{b})$. The latter polynomials satisfy the
  difference equation in $x$ \cite[18.22.13--18.22.15]{DLMF},
  \begin{equation}\label{H.1} 
  A(x) p_n(x+i) - (A(x) + C(x) - n (n + 2 {\rm Re} \, (a+b) - 1) ) p_n(x) + C(x) p_n(x-i) = 0,
  \end{equation}
  where
 \begin{equation}\label{H.2}  
 A(x) = (x + i \bar{a}) (x + i \bar{b}), \qquad  C(x) = (x - i {a}) (x - i {b}).
  \end{equation} 
  
  \begin{proposition}
  Let $ \tilde{\mu}_0^{(J)}$ be given by (\ref{E.7c}). For general ${\rm Re} \, k > -1/2$ and $\beta = 2$, we have
    \begin{equation}\label{H.4}
  \tilde{\mu}_k^{(J)}    =   { \tilde{\mu}_0^{(J)}  i^{1-N} \over N (3/2 - (a + N))_{N-1}} 
   s_{N-1} \Big (i(k+1);1,{1 \over 2} - (a + N) \Big ),
   \end{equation}
   and thus $ {\mu}_k^{(J)}$ with $\beta = 2$ is given in terms of the continuous Hahn polynomials according to
     \begin{equation}\label{H.4a}
 {\mu}_k^{(J)}    =    {\tilde{\mu}_0^{(J)} i^{1-N}  \Gamma(k+1/2) \over  N \Gamma(k +  N  + a + 3/2)  (3/2 - (a + N))_{N-1}} 
   s_{N-1} \Big (i(k+1);1,{1 \over 2} - (a + N) \Big ).
   \end{equation}

   \end{proposition}
  
  \begin{proof}
  Comparing (\ref{H.1}), (\ref{H.2}) with (\ref{E.7b}) we see that
  \begin{equation}\label{H.3}
 C_{N,a}  s_{N-1} \Big (i(k+1);1,{1 \over 2} - (a + N) \Big ),
    \end{equation} 
    where $C_{N,a}$ is independent of $k$,
    satisfies the recurrence (\ref{E.7b}). This is also a polynomial of degree $(N-1)$ in $k$,
    which upon choosing 
   \begin{equation}
    C_{N,a} =  { \tilde{\mu}_0^{(J)}  \over s_{N-1}\Big (i;1,{1 \over 2} - ( a  + N) \Big )} = 
     { \tilde{\mu}_0^{(J)}  i^{1-N} \over N (3/2 - (a + N))_{N-1}} ,
  \end{equation}      
  where the second equality follows from (\ref{H}), agrees with $ \tilde{\mu}_k^{(J)}$ for
  $k = 0$.
   The recurrence (\ref{E.7b}) then gives that this polynomial is equal to $\tilde{\mu}_k^{(J)}$
   at each positive
    integer value of $k$, and thus the polynomials must in fact be identical, giving (\ref{H.4}).
    
    The formula (\ref{H.4a}) follows from (\ref{H.4}) using (\ref{S.2b}) with $\beta = 2$.
    \end{proof}
    
    \begin{remark}
    For $\beta$ even, we know from the discussion leading to (\ref{S.2b}) that $\mu_k^{(J)}$ permits a
    factorisation involving a polynomial part $\tilde{\mu}_k^{(J)}$. For $\beta = 2$, this polynomial is given
    in terms of particular continuous Hahn polynomials according to (\ref{H.4}).
    As noted in \cite{ABGS20} in the context of the moments of the symmetric Cauchy ensemble (the exact evaluation of
  $ \mu_k^{(Cy)}$ in terms of continuous Hahn polynomials implied by \eqref{4.8}, (\ref{H.4a}) agrees
  with that given therein), the fact that $s_n(x;a,b)$ has all real zeros implies that the zeros of
    $\tilde{\mu}_k |_{\beta = 2}$ are all on the line ${\rm Re} \, k = - 1$. However, this property does not
    carry over to other values of $\beta$. Thus, it is a simple exercise to compute explicitly $\tilde{\mu}_k^{(J)}$
    with $\beta$ even in the case $N = 2$. Further specialising to $\beta = 4$, when the polynomial is a
    quadratic, it is found that the zeros are distinct, and lie on the negative real axis at positions varying with $a$.
    \end{remark}

  \subsection{Resolvent for the symmetric Cauchy case}
Accompanying the definition \eqref{4.7} of the sum of successive even moments $\mu_k^{(Cy)}$ of the Cauchy ensemble is the requirement that $-1/2 < {\rm Re} \, k < \alpha - 1/2$. This implies that $\mu_k^{(Cy)}$, and thus $m_k^{(Cy)}$, is well defined for only finitely many positive integers $k$, which, in turn, prevents us from making sense of their generating functions. This complication vanishes in the large $N$ limit when we set $\alpha=\hat{\alpha}\beta N/2$ as seen in \S\ref{s3.1.3}. We divide this subsection into three. In the first part, we study the generating functions of the large $N$ limiting forms of the $\mu_k^{(Cy)}$ and $m_k^{(Cy)}$. We then study the finite-$N$ analogues of these generating functions, which are to be understood in a formal sense through analytic continuation. In the third and final part, we use the differential equations of Section \ref{S3} to study the $1/N$ expansions of the aforementioned formal generating functions.

  \subsubsection{Large $N$ limit}
Here, we will make a study of the moments in the large $N$ limit to give a complementary viewpoint to the results
of \S\ref{s3.1.3}. For this purpose, define
  $$
  \hat{\mu}_k^{(Cy)} = \lim_{N \to \infty} {1 \over N} \mu_k^{(Cy)}  \Big |_{\alpha = \hat{\alpha}\beta N/2}, \qquad
  \hat{m}_k^{(Cy)} = \lim_{N \to \infty} {1 \over N} m_{2k}^{(Cy)}  \Big |_{\alpha = \hat{\alpha}\beta N/2} 
  $$
  and introduce the generating functions
   \begin{equation}\label{H1}
   H(x) = \sum_{k=0}^\infty  { \hat{\mu}_k^{(Cy)} \over x^{2k}} , \qquad G(x) = \sum_{k=0}^\infty   {\hat{m}_{2k}^{(Cy)} \over x^{2k}}.
  \end{equation}
  The definition (\ref{4.7}) shows that the latter are related by
    \begin{equation}\label{H2}
    G(x) =  {H(x) + x^2 \over 1 + x^2}.
  \end{equation}
  Furthermore, with $\hat{\rho}_{(1)}^{(Cy)}$ denoting the LHS of (\ref{9.3z}), and $I$ denoting its support, for $x \notin I$, we have
\begin{equation}\label{H3a}  
G(x) = \int_I {\hat{\rho}_{(1)}^{(Cy)}(y) \over 1 - (y/x)^2} \, dy,
 \end{equation}
 and so by the Sokhotski-Plemelj formula,
\begin{equation}\label{H3b}  
s \hat{\rho}_{(1)}^{(Cy)}(s) = {1 \over \pi} \lim_{\epsilon \to 0^+} {\rm Im} \, G(x) \Big |_{x = s - i \epsilon}.  
 \end{equation}

We will now show how to use (\ref{H3b}) to re-derive (\ref{9.3z}). To begin, 
with $\alpha = \hat{\alpha} N$ in (\ref{4.21}), (\ref{4.21a}), equating leading powers in $N$ shows
\begin{equation}\label{H4} 
- (2k + 4) \hat{\alpha}^2  \hat{\mu}_{k+1}^{(Cy)} + (2k + 1) (1 + 2  \hat{\alpha}) \hat{\mu}_{k}^{(Cy)} = 0
\end{equation}
subject to the initial condition
\begin{equation}\label{H4a} 
 \hat{\mu}_{0}^{(Cy)} = {1 + 2 \hat{\alpha} \over 2 \hat{\alpha} }.
\end{equation}
Hence, by iterating (\ref{H4}), it follows that
\begin{equation}\label{H4b} 
  \hat{\mu}_{k}^{(Cy)} = {  (1 + 2  \hat{\alpha})^{k+1} \over 2  \hat{\alpha}^{2k+1}} {(1/2)_k \over (2)_k} = {  (1 + 2  \hat{\alpha})^{k+1} \over (2  \hat{\alpha})^{2k+1}} C_k,
\end{equation}
where $(u)_k$ is the (rising) Pochhammer  symbol, and $C_k$ denotes the $k$-th Catalan number (recall (\ref{1.1})). This substituted in the first equation of (\ref{H1}) tells
us that
\begin{align}\label{H4c} 
H(x) = {  (1 + 2  \hat{\alpha} ) \over 2  \hat{\alpha} }\, {}_2 F_1 (1, 1/2; 2; (1 + 2  \hat{\alpha} )/(\hat{\alpha} x)^2) 
 =  {  (1 + 2  \hat{\alpha} ) \over   \hat{\alpha} } {1 - \sqrt{1 - z} \over z} \Big |_{z = (1 + 2  \hat{\alpha} )/(\hat{\alpha} x)^2}.
 \end{align}
Substituting in (\ref{H2}) and then substituting the result in (\ref{H3b}), we reclaim (\ref{9.3z}). It can be checked that \eqref{H4b} satisfies \eqref{4.20} with $\alpha=\hat{\alpha}\beta N/2$, owing to  $\hat{\rho}^{(Cy)}_{(1)}(s)$ then
being $\beta$ independent.

 \subsubsection{Finite $N$}
At this point, let us stress that the $\mu_k^{(Cy)}$ and their recurrences have formal interpretations through analytic continuation without the need for taking $N$ large. Hence, we are able to complement the above computation with the $1/N$ correction to $\hat{\rho}^{(Cy)}_{(1)}(s)$. To this end, introduce the formal sums
\begin{equation} \label{4.41}
W_1(x)=\sum_{k=0}^{\infty}\frac{m_{k}^{(Cy)}}{x^{k+1}},\qquad\hat{W}_1(x)=W_1(x) \Big|_{\alpha=\hat{\alpha}\beta N/2}
\end{equation}
so that for $\beta=2$, $\lim_{N\to\infty}x\hat{W}_1(x)=G(x)$. It is known for the Gaussian, Laguerre and Jacobi (with $(0,1)$ as support) ensembles that the corresponding resolvents defined analogously to $W_1(x)$ above satisfy the same differential equations as the densities $\rho_{(1)}(x)$ with additional inhomogeneous terms. The same is true in the Cauchy case, seen by a tweaking of the reasoning in \cite[Appendix A]{RF19}. Thus, replacing $\rho^{(Cy)}_{(1)}(x)$ in the differential equation \eqref{3.7} by the expression for $W_1(x)$ given in \eqref{4.41} yields inhomogeneous differential equations satisfied by $W_1(x)$.

\begin{proposition} \label{P4.9}
In the setting of Proposition \ref{P3.3}, we have
\begin{equation} \label{4.42}
\mathcal{D}^{(Cy)}_{\beta,N}\,\frac{1}{N}W_1(x)=\begin{cases} 4(N+\alpha)(N+2\alpha),&\beta=2, \\ h(x;N,\alpha),&\beta=4, \\ h(x;-N/2,-2\alpha),&\beta=1,\end{cases}
\end{equation}
where
\begin{multline}
h(x;N,\alpha)=8(2N+2\alpha-1)\Big[(2N+1)(8N^2+x^2-1)+4N\alpha(6N+x^2+3) \\
+\alpha^2(3x^2-4N(x^2-2)+5)-2\alpha(1+x^2\alpha^2)\Big].
\end{multline}
\end{proposition}
Comparing to Corollary \ref{C4.3}, one should interpret $\mathcal{D}_{\beta,N}^{(Cy)}$ as encoding the moment recurrences therein, while the right-hand side of \eqref{4.42} encodes the initial conditions.

 \subsubsection{Expansion in  $1/N$}
Considering the moment recurrences and initial conditions of Corollary \ref{C4.3} directly, it is evident that $m_k^{(Cy)}/N$ admits a $1/N$ expansion upon setting $\alpha=\hat{\alpha}\beta N/2$. Thus, we may take as ansatz
\begin{equation}\label{4.44}
\frac{1}{N}\hat{W}_1(x)=\sum_{l=0}^{\infty}\frac{\hat{W}_{1,l}(x)}{N^l}.
\end{equation}
Setting $\alpha=\hat{\alpha}\beta N/2$ in \eqref{4.42} results in differential equations for $\hat{W}_1(x)/N$. Substituting the above ansatz into these differential equations and equating powers of $N$ then gives first order differential equations for the expansion coefficients $\hat{W}_{1,l}(x)$.

\begin{proposition} \label{P4.10}
Set $\beta=2$. The leading term $\hat{W}_{1,0}(x)$ in (\ref{4.44}) satisfies the differential equation
\begin{equation} \label{4.45}
(1+x^2)(\hat{\alpha}^2x^2-2\hat{\alpha}-1)\hat{W}_{1,0}'(x)+(\hat{\alpha}^2x^2-\hat{\alpha}^2-4\hat{\alpha}-2)x\hat{W}_{1,0}(x)=-(\hat{\alpha}+1)(2\hat{\alpha}+1).
\end{equation}
Next, $\hat{W}_{1,1}(x)$ satisfies the homogeneous differential equation corresponding to \eqref{4.45}. 
Finally, for $l\geq2$, we have
\begin{multline} \label{4.47}
4(1+x^2)(\hat{\alpha}^2x^2-2\hat{\alpha}-1)\hat{W}_{1,l}'(x)+4(\hat{\alpha}^2x^2-\hat{\alpha}^2-4\hat{\alpha}-2)x\hat{W}_{1,l}(x) \\
=(1+x^2)^3\hat{W}_{1,l-2}'''(x) +8(1+x^2)^2x\hat{W}_{1,l-2}''(x) \\
+2(1+x^2)(7x^2+3)\hat{W}_{1,l-2}'(x)+4(1+x^2)x\hat{W}_{1,l-2}(x).
\end{multline}
\end{proposition}

\begin{remark} \label{R4.11}
1.~The equation (\ref{4.45}) has general solution
\begin{equation} \label{4.47x}
\hat{W}_{1,0}(x)=\frac{(\hat{\alpha}+1)x}{1+x^2}+C\frac{\sqrt{\hat{\alpha}^2x^2-2\hat{\alpha}-1}}{1+x^2}.
\end{equation}
The integration constant $C$ is to be set equal to $-1$. Then $\hat{W}_{1,0}(x)\sim1/x$ as $x\to\infty$, in line with $m_0^{(Cy)}=N$. This shows $x\hat{W}_{1,0}(x)=G(x)$, as expected. \\
2.~The fact that $\hat{W}_{1,1}(x)$ satisfies the homogeneous differential equation corresponding to \eqref{4.45}
 is consistent with the requirement that $\hat{W}_{1,l}(x;\beta=2)=0$ for odd $l$, since the moments $m_k^{(Cy)}$ are odd functions of $N$ when $\beta=2$. \\
 3.~To uniquely determine $\hat{W}_{1,2p}(x)$ for $p\geq1$ an integer, we use (\ref{4.21a}) to compare the $1/N^2$ expansion of $\frac{1}{N}m_2^{(Cy)}|_{\alpha=\hat{\alpha}N}$ to \eqref{4.44} and consequently see that
\begin{equation}
\hat{W}_{1,2p}(x)\underset{x\to\infty}{\sim}\frac{1+2\hat{\alpha}}{(2\hat{\alpha})^{2p+1}x^3}.
\end{equation}
Thus, the integration constant present in the solution of \eqref{4.47} is zero. In particular, solving (\ref{4.47}) with $l=2$ gives
\begin{equation} \label{4.49}
\hat{W}_{1,2}(x)=\frac{\hat{\alpha}^2(1+2\hat{\alpha})(1+x^2)}{8(\hat{\alpha}^2x^2-2\hat{\alpha}-1)^{5/2}}.
\end{equation}
Hence by the Sokhotski-Plemelj formula,
\begin{equation} \label{4.50}
\frac{1}{N}\rho_{(1)}^{(Cy)}(x)\Big |_{\alpha = \hat{\alpha} N}=\hat{\rho}_{(1)}^{(Cy)}(x)+\frac{1}{N^2}\frac{\hat{\alpha}^2(1+2\hat{\alpha})(1+x^2)}{8\pi(1+2\hat{\alpha}-\hat{\alpha}^2x^2)^{5/2}}\chi_{|x|<\sqrt{1+2\hat{\alpha}}/\hat{\alpha}}+{\rm O}\left(\frac{1}{N^4}\right).
\end{equation}
\end{remark}

In the cases $\beta=1$ and $4$, the analogue of Proposition \ref{P4.10} is extracted from equation \eqref{4.42} using the procedure described above. We do not present it here for brevity, but simply note that in these cases, the first order differential equation satisfied by $\hat{W}_{1,l}(x)$ has inhomogeneous terms dependent on $\hat{W}_{1,p}(x)$ for $1\leq p\leq 4$. This is in contrast to the $\beta=2$ case, where the differential equation for $\hat{W}_{1,l}(x)$ does not have $\hat{W}_{1,l-1}(x)$ present in the inhomogeneous term. Skipping these details, we now supplement equations \eqref{4.49} and \eqref{4.50} with their $\beta=1,4$ analogues.

\begin{proposition} \label{P4.12}
Set $\beta=1$ or $4$. Then, $\hat{W}_{1,0}(x)$ is given by \eqref{4.47x} with $C=-1$ (recall the $\beta$-independence of $\hat{\rho}_{(1)}^{(Cy)}(x)$). In addition,
\begin{align}
\hat{W}_{1,1}(x)&=\frac{(1-2/\beta)\hat{\alpha}}{2}\left[\frac{1}{\sqrt{\hat{\alpha}^2x^2-2\hat{\alpha}-1}}-\frac{\hat{\alpha}x}{\hat{\alpha}^2x^2-2\hat{\alpha}-1}\right], \label{4.51}
\\ \hat{W}_{1,2}(x)&=\frac{(1-2/\beta)^2}{2}\left[\frac{\hat{\alpha}^2(1+2\hat{\alpha}+\hat{\alpha}^2)x^2}{(\hat{\alpha}^2x^2-2\hat{\alpha}-1)^{5/2}}+\frac{\hat{\alpha}^2(1+2\hat{\alpha})(1+x^2)}{4(\hat{\alpha}^2x^2-2\hat{\alpha}-1)^{5/2}}-\frac{\hat{\alpha}(1+2\hat{\alpha}+\hat{\alpha}^2)x}{(\hat{\alpha}^2x^2-2\hat{\alpha}-1)^{2}}\right] \nonumber
\\&\qquad+\frac{\hat{\alpha}^2(1+2\hat{\alpha})(1+x^2)}{4\beta(\hat{\alpha}^2x^2-2\hat{\alpha}-1)^{5/2}}. \label{4.52}
\end{align}
Applying the Sokhotski-Plemelj formula shows
\begin{multline} \label{4.53}
\frac{1}{N}\rho_{(1)}^{(Cy)}(x)\Big|_{\alpha=\hat{\alpha}\beta N/2}=\hat{\rho}_{(1)}^{(Cy)}(x)+\frac{1}{N}(1-2/\beta)\left[\frac{\hat{\alpha}}{2\pi\sqrt{1+2\hat{\alpha}-\hat{\alpha}^2x^2}}\chi_{|x|<\sqrt{1+2\hat{\alpha}}/\hat{\alpha}}\right. \\
\left.-\frac{1}{4}\delta(x-\sqrt{1+2\hat{\alpha}}/\hat{\alpha})-\frac{1}{4}\delta(x+\sqrt{1+2\hat{\alpha}}/\hat{\alpha})\right]+{\rm O}\left(\frac{1}{N^2}\right),
\end{multline}
where $\delta(x)$ is the Dirac delta.
\end{proposition}

Although derived for $\beta=1, 4$, Remark \ref{R4.11} tells us that both \eqref{4.51} and \eqref{4.52} are also valid for $\beta = 2$. In fact, adopting
the viewpoint of a loop equation analysis (see \cite[\S 3.1]{WF14}, \cite{FRW17}), it is expected that $\beta^p \hat{W}_{1,2p}$ is
an even polynomial in the variable $ (\sqrt{\beta} - 2/\sqrt{\beta})$. If this was to be assumed, it would follow that the
results of Proposition \ref{P4.12} are in fact valid for general $\beta > 0$.

More parallels between the Cauchy ensemble resolvent and their Gaussian and Laguerre ensemble analogues can be drawn from the references \cite{WF14}, \cite{FRW17}. For one, the structure exhibited in \eqref{4.50} of the $1/N^2$ correction term having
a nonintegrable singularity at each endpoint of the support, which diverges
like an inverse $5/2$ power, is shared by the analogous expansion for the GUE and
the LUE (the latter with the Laguerre parameter proportional to $N$). Moreover, when $\beta\neq2$,
the correction term given in these references is of order $1/N$ with Dirac delta contributions at the endpoints of support.
This property is shared with the correction term in \eqref{4.53}. As a final remark, we note that the aforementioned references also
show how integration of the correction terms against monomials, so as to
compute the corresponding moments, can be regularised
using integration by parts.

     \subsection*{Acknowledgements}
This research is part of the program of study supported by the Australian Research Council Centre of Excellence ACEMS. The work of PJF was also partially supported by the Australian Research Council Grant DP210102887, and that of AAR by the Australian Government Research Training Program Scholarship.
    
\providecommand{\bysame}{\leavevmode\hbox to3em{\hrulefill}\thinspace}
\providecommand{\MR}{\relax\ifhmode\unskip\space\fi MR }
\providecommand{\MRhref}[2]{%
  \href{http://www.ams.org/mathscinet-getitem?mr=#1}{#2}
}
\providecommand{\href}[2]{#2}

 \end{document}